\documentclass[12pt]{article}

\usepackage[a4paper,margin=1in]{geometry}
\usepackage{amsmath,amssymb,mathtools}
\usepackage{amsthm}

\usepackage[hidelinks]{hyperref}
\usepackage{enumitem}
\usepackage{booktabs}
\usepackage[most]{tcolorbox}
\usepackage{array}
\usepackage{tabularx}
\usepackage{caption}

\usepackage{tikz}
\usetikzlibrary{arrows.meta,positioning,fit}
\usepackage{parskip}      
\usepackage{pdflscape}
\usepackage{ragged2e}

\usepackage{tikz}
\usetikzlibrary{arrows.meta,positioning,fit,calc}

\usepackage{tcolorbox}

\newcommand{\K}{K}
\date{\today}

\newtheoremstyle{myplain}%
  {6pt}{6pt}{\itshape}{}{\bfseries}{.}{.5em}{}
\newtheoremstyle{mydefinition}%
  {6pt}{6pt}{\normalfont}{}{\bfseries}{.}{.5em}{}
\newtheoremstyle{myremark}%
  {6pt}{6pt}{\normalfont}{}{\itshape}{.}{.5em}{}

\theoremstyle{myplain}
\newtheorem{theorem}{Theorem}[section]
\newtheorem{lemma}{Lemma}[section]

\theoremstyle{mydefinition}
\newtheorem{definition}{Definition}[section]

\theoremstyle{myremark}

\tcbset{
  theorem box/.style={
    enhanced, breakable,
    colback=white,
    colframe=black!25,
    boxrule=0.6pt,
    arc=1.5pt,
    left=8pt, right=8pt, top=6pt, bottom=6pt
  },
  definition box/.style={
    enhanced, breakable,
    colback=white,
    colframe=black!15,
    boxrule=0.6pt,
    arc=1.5pt,
    left=8pt, right=8pt, top=6pt, bottom=6pt
  },
  remark box/.style={
    enhanced, breakable,
    colback=white,
    colframe=black!10,
    boxrule=0.5pt,
    arc=1.5pt,
    left=8pt, right=8pt, top=6pt, bottom=6pt
  }
}

\tcolorboxenvironment{theorem}{theorem box}
\tcolorboxenvironment{lemma}{theorem box}
\tcolorboxenvironment{proposition}{theorem box}
\tcolorboxenvironment{corollary}{theorem box}
\tcolorboxenvironment{conjecture}{theorem box}

\tcolorboxenvironment{definition}{definition box}
\tcolorboxenvironment{example}{definition box}

\tcolorboxenvironment{remark}{remark box}

\tcbset{colback=white,colframe=black!20!black,boxsep=4pt,arc=1pt,left=6pt,right=6pt,top=6pt,bottom=6pt}
\newenvironment{plainbox}{\begin{tcolorbox}}{\end{tcolorbox}}

\title{The Algorithmic Regulator }
\author{Giulio Ruffini\footnote{giulio.ruffini@bcom.one, giulio.ruffini@neuroelectrics.com}}
\date{Oct 14, 2025}

\begin{document}

\maketitle

\begin{abstract}
The regulator theorem states that, under certain conditions, any optimal controller must embody a model of the system it regulates, grounding the idea that controllers embed, explicitly or implicitly, internal models of the controlled. This principle underpins neuroscience and predictive brain theories like the Free-Energy Principle or  Kolmogorov/Algorithmic Agent theory.  However, the theorem is only proven in limited settings.  
Here, we treat the deterministic, closed, coupled world-regulator system ($W$,$R$) as a single self-delimiting program $p$ via a constant-size wrapper that produces the world output string~$x$ fed to the regulator. We analyze regulation from the viewpoint of the algorithmic complexity of the output, $K(x)$ (regulation as compression).
We define \(R\) to be a \emph{good algorithmic regulator}  if it  \emph{reduces} the algorithmic complexity of the readout relative to a null (unregulated) baseline $\varnothing$, i.e., 
$
\Delta =K\big(O_{W,\varnothing}\big)-K\big(O_{W,R}\big)>0.
$
We then prove that the larger \(\Delta\) is, the more world-regulator pairs with high mutual algorithmic information are favored.  More precisely, a complexity gap \(\Delta>0\) yields
\(\Pr((W,R)\mid x)\le C\,2^{\,M(W{:}R)}\,2^{-\Delta}\), making low \(M(W{:}R)\) exponentially unlikely as \(\Delta\) grows.  
This is an AIT version of the idea that “the regulator contains a model of the world.” The framework is distribution‑free, applies to individual sequences, and complements the Internal Model Principle.
Beyond this necessity claim, the same coding‑theorem calculus singles out a \emph{canonical scalar objective} and implicates a \textit{planner}. On the realized episode, a regulator behaves \emph{as if} it minimized the conditional description length of the readout.  
\end{abstract}

 \clearpage\tableofcontents \clearpage

\section{Introduction}
\label{sec:intro}

In the Kolmogorov Theory (KT) of consciousness, an \emph{algorithmic agent} is a system that maintains (tele)homeostasis (persistence of self or kind) by learning and running succinct generative models of its world coupled to an objective function and a action planner \cite{ruffiniAlgorithmicInformationTheory2017,ruffiniAITFoundationsStructured2022,ruffiniStructuredDynamicsAlgorithmic2025}. Closely related, \emph{Active Inference (AIF)} models biological agents as minimizing variational free energy under a generative model \cite{karlFreeEnergyPrinciple2012,parr2022active}. These frameworks suggest that “agents with world‑modeling engines, objective functions, and planners” are natural minimal models of homeostasis (goal‑conditioned setpoint control). But for the kinds of homeostatic systems we actually encounter in nature (cells, organisms, engineered servos), how can we tell—operationally—whether they are algorithmic agents in this sense? 

The classical cybernetics statement that ``every good regulator of a system must be a model of that system'' originates with Conant and Ashby’s 1970 paper (the \emph{Good Regulator Theorem}, GRT) \cite{ConantAshby1970}.  While influential, the GRT has been criticized for the looseness of its definitions of ``model'' and ``goodness'', and for a proof that does not clearly deliver the headline claim~ \cite{BaezGRT}.  In modern control theory, the rigorous statement that fills a similar conceptual niche is the \emph{Internal Model Principle} (IMP): under appropriate hypotheses, perfect regulation or disturbance rejection for a given signal class requires that the controller embed a dynamical copy of the signal generator \cite{FrancisWonham1975AMO,FrancisWonham1976Automatica,Sontag2003}.  The IMP is precise (and falsifiable) within its scope, and is now a standard backbone for robust control; see \cite{BinEtAl2022-AR} for a contemporary review across control, bioengineering, and neuroscience.
However, the classical IMP is a linear result: for finite-dimensional LTI plants ( linear, \emph{time-invariant} meaning the system matrices do not change with time) and exogenous signals generated by a finite-dimensional LTI exosystem, robust asymptotic tracking/disturbance rejection requires that the controller embed a copy of the exosystem dynamics  \cite{FrancisWonham1976}. For nonlinear systems, the appropriate generalization is the nonlinear output-regulation framework: if the regulator equations admit smooth solutions and the plant's zero dynamics on the regulated manifold are (locally) stable, together with suitable immersion/detectability assumptions, then one can construct dynamic output-feedback regulators that embed a (possibly adaptive) internal model and achieve local or semiglobal robust regulation \cite{Isidori45168,huangNonlinearOutputRegulation2004,priscoliAdaptiveObserversNonlinear2006}. However, absent these structural hypotheses, a complete nonlinear analogue of the IMP with the same necessity/robustness guarantees as in the LTI case is not generally available. Table~\ref{tab:grt-imp-agrt} provides a comparison of the different regulator theorem statements, which can be compared with the one presented here.

In this paper, we recast the modeling requirement in a setting independent of linearity, probability, exact regulation or specific signal classes, by using \emph{algorithmic information theory} (AIT).  We model a world \(W\) and a regulator \(R\) as deterministic causal Turing machines that interact over interface tapes.  We denote the world output by \(x=O_W\) (over some temporal horizon of length $N$).      
Our main technical claim is that regulation in the algorithmic sense, i.e., simplicity, \emph{forces} algorithmic dependence between \(W\) and~\(R\).  


\subsubsection*{Definition of model}
A  \emph{model} in the present context is a program capable of compressing (or generating) data. Similarly, “the regulator contains a model of the world” is interpreted in an \emph{algorithmic‑information} sense: the regulator $R$ carries nontrivial information about $W$, quantified by positive mutual algorithmic information $M(W{:}R)>0$ (up to the standard $O(\log)$ slack). Equivalently, knowing $R$ makes the shortest description of $W$ strictly shorter, $K(W\mid R)<K(W)$. This notion does \emph{not} require $R$ to embed a dynamical copy of $W$; rather, it formalizes “model content” as mutual algorithmic information. 

We formalize this  with the following definition:
\begin{definition}[Algorithmic “internal model”]
Given a fixed horizon $N$ (implicitly conditioned), we say that $R$ \emph{contains an internal model of} $W$ \emph{in the algorithmic sense} if $M(W{:}R)>0$ (up to $O(\log)$), equivalently $K(W\mid R)<K(W)$. The magnitude of $M(W{:}R)$ quantifies the amount of computable structure in $W$ that $R$ carries.
\end{definition}
The definition ground on mutual algorithmic information $M$  is further motivated by the following:
(i) \textit{Machine invariance:} $M$ is invariant up to $O(1)$ under changes of universal machine. (ii) \textit{Distribution‑free:} $M$ is defined for individual objects (programs), not probabilistic models. (iii) \textit{Operational meaning:} $M(W{:}R)$ is precisely the codelength reduction in describing $W$ when $R$ is known, aligning with MDL/Occam reasoning via the Coding Theorem  \cite{Solomonoff64a,ZvonkinLevin70,LiVitanyi4th}. 


This is the appropriate lens for our contrastive results, and it complements the Internal Model Principle, where “model” means a dynamical replica of the Exosystem (a part of the World in our framework, see Figure~\ref{fig:wr-imp-yourlayout}) under stated structural hypotheses \cite{FrancisWonham1976,Isidori45168,huangNonlinearOutputRegulation2004,priscoliAdaptiveObserversNonlinear2006}. Conceptually, our AIT result is complementary to the IMP: whereas the IMP states \emph{what} structural content must be present in a controller to achieve perfect regulation for a given signal class \cite{FrancisWonham1975AMO,FrancisWonham1976Automatica,Sontag2003}, our results quantify \emph{how much} algorithmic information the regulator must carry about the world whenever it succeeds in making the measured outcome compressible.

\subsubsection*{Regulation as compression}
We score \emph{regulation} by how \emph{compressible} a task‑weighted error stream is. Let $x_t$ be the weighted error and $x_{1:T}$ the $T$‑sample string. Fix a prefix‑free lossless code (e.g., a universal compressor) and define the per‑sample codelength $L_T := \tfrac{1}{T} L_C(x_{1:T})$. A regulator $R$ is \emph{better} on horizon $T$ when it makes $L_T$ smaller than a null baseline $\varnothing$, i.e., when the \emph{contrastive gap} $\Delta := L_T(x;\varnothing) - L_T(x;R)$ is positive. 
This choice is natural: for stationary ergodic data, normalized universal codelengths converge (a.s.\footnote{``Almost surely'': with probability $1$.}) to the Shannon entropy rate $h(x)$, and (under standard computability assumptions) $K(x_{1:T})/T = h(x) + o(1)$ almost surely; thus the Kolmogorov‑based criterion reduces to the Shannon criterion when those stochastic assumptions hold—while remaining meaningful outside them \cite{ZivLempel1977,CoverThomas2006,LiVitanyi4th}.

To see the connection between regulation and compression in more detail, let $h_{x_{1:T}}(\alpha):=\min_{S\ni x_{1:T},\,K(S)\le \alpha}\log|S|$ denote the Kolmogorov \emph{structure function} \cite{VereshchaginVitanyi2004}. Regulation amounts to moving \emph{down} this curve: as the regulator invests model bits (larger $\alpha$), more regularity in $x_{1:T}$ is captured and the residual randomness $h_{x_{1:T}}(\alpha)$ drops, approaching $0$ at perfect regulation.
 The notion of robability emerges along this path. Replacing set‑models by probabilistic models $\{P_M\}$ turns the two‑part description into the standard MDL form
\[
L(x_{1:T};M)\;\approx\;K(M)\;-\!\sum_{t=1}^T \log P_M(x_t),
\]
where the second term is the ideal codelength under $P_M$ (Shannon coding: $-\log P_M$) \cite{CoverThomas2006,LiVitanyi4th,Grunwald2007}. If the regulator must hedge over multiple models, the \emph{mixture/Bayes code} with prior $\pi$ uses $\bar P(x_{1:T})=\sum_M \pi(M) P_M(x_{1:T})$ and assigns
\[
L_{\rm mix}(x_{1:T}) \;=\; -\log \bar P(x_{1:T}) \;=\; -\log\!\sum_M \pi(M)P_M(x_{1:T}),
\]
a valid prefix code whose regret relative to the best single model $M^\star$ is bounded by $-\log \pi(M^\star)$; with $\pi(M)\!\propto\!2^{-K(M)}$ (Solomonoff/Occam), the penalty matches the model description length \cite{BarronRissanenYu1998,Grunwald2007,solomonoffFormalTheoryInductive1964a,Solomonoff64a,LiVitanyi4th}. Thus, \emph{probabilistic/Bayesian regulation is the coding‑optimal way to descend $h_{x_{1:T}}(\alpha)$}, aligning with the multi‑model argument in \cite{Ruffini2024Emergence}.

Finally, we can treat the regulator input as an error signal quantized at fixed sensor resolution; the per‑sample codelength of $x_{1:T}$ under a universal compressor converges to the entropy rate for stationary sources. For Gaussian processes,
\[
h(x)\;=\;\frac{1}{4\pi}\!\int_{-\pi}^{\pi}\log\!\big(2\pi e\,S_{xx}(\omega)\big)\,d\omega,
\]
where \(S_{xx}\) is the input power spectral density~\cite{ZivLempel1977,Gray2011} of the error signal $x$, so attenuating in‑band sensitivity (reducing $S_{xx}$ where it matters) \emph{reduces} codelength \cite{Gray2011}. In the scalar white‑Gaussian case with variance $\sigma^2$, $h=\tfrac{1}{2}\log(2\pi e\,\sigma^2)$, so smaller‑amplitude fluctuations (smaller $\sigma$) mean lower entropy and better compressibility. In short, ``compressible error'' matches the classical view: good regulation removes variability/uncertainty in the task band and accords with the IMP \cite{AstromMurray2008,FrancisWonham1976}.

In the next sections, we first provide an overview of the AIT setting and of the results, followed by the analysis of the single episode scenario. The next section provides a formal definition of the algorithmic regulator and the corresponding theorem.

\section{Setting}
 
 
Unless stated otherwise, $U$ is the standard \emph{three–tape} universal \emph{prefix} Turing machine: 
a read‑only input tape holding a self‑delimiting program $p$, a work tape (private scratch memory), 
and a write‑only output tape. When we write $U(p)=x$ we mean that, upon halting, the contents of the 
\emph{output} tape equal $x$; the work tape is never part of the scored output. The domain of halting 
programs is prefix‑free, so Kraft–McMillan applies and the universal a priori semimeasure 
$m(x)=\sum_{U(p)=x}2^{-|p|}$ is well defined. By the invariance theorem, replacing $U$ by any other 
universal prefix machine (single‑ or multi‑tape) changes all complexities only by an additive $O(1)$; 
all Coding‑Theorem statements we use depend only on prefix‑freeness and therefore remain valid up to 
these constants (see, e.g., \cite{LiVitanyi4th}).

The (prefix) \textit{Kolmogorov complexity }of $x$ is the length of its shortest description,
\[
\K(x)\ :=\ \min\{\,|p|:\ U(p)=x\,\}.
\]
Intuitively, $\K(x)$ is the best achievable compressed size of $x$ on $U$. If $\K(x)\ll |x|$, then $x$ has a short generative regularity; if $\K(x)\approx |x|$, $x$ is (algorithmically) random. By the invariance theorem, $\K$ is machine‑independent up to an additive constant $O(1)$ \cite{LiVitanyi4th}. A fundamental limitation is that $\K(\cdot)$ is \emph{not computable}: no algorithm can output $\K(x)$ for all $x$ \cite{chaitinTheoryProgramSize1975,LiVitanyi4th}. However, algorithms for upper bounds of $K(x)$ exist, as we discuss below.

Given auxiliary data $y$ on a read‑only \emph{auxiliary} tape, the \textit{conditional complexity}
\[
\K(x\mid y)\ :=\ \min\{\,|p|:\ U(p,y)=x\,\}
\]
is the shortest description of $x$ \emph{given} $y$. It operationalizes how much new information is needed to reconstruct $x$ once $y$ is known (e.g., “world given regulator,” or “output given model”).

The \textit{mutual algorithmic information} (up to the usual $O(\log)$ slack) is
\[
M(x{:}y)\ :=\ \K(x)+\K(y)-\K(x,y)\ =\ \K(x)-\K(x\mid y)\ =\ \K(y)-\K(y\mid x)\ \pm O(\log).
\]
$M(x{:}y)$ measures the algorithmically \emph{shared} structure between $x$ and $y$: how many bits we save when describing one with the help of the other. In our setting, “the regulator contains a model of the world” means $M(W{:}R)>0$ (information‑theoretic dependence), not necessarily a dynamical replica.

Intuitively, strings produced by shorter programs are more likely. Solomonoff–\textit{Levin’s universal a priori semimeasure} $m(x)$ and the\textit{ Coding Theorem} link probability and description length:
\begin{equation} \label{eq:levin}
-\log_2 m(x)\ =\ \K(x)\ \pm O(1),
\end{equation}
providing a universal Occam calculus over individual strings.   \cite{Solomonoff64a,solomonoffFormalTheoryInductive1964a, zvonkinCOMPLEXITYFINITEOBJECTS1970,LiVitanyi4th}.

In what follows, a finite temporal horizon $N$ is fixed throughout; unless stated otherwise, we implicitly
condition on $N$ (e.g., write $\K(x)$ for $\K(x\mid N)$). All $O(1)$ constants depend only on the choice of $U$
(and the fixed constant‑overhead wrapper that decodes $(W,R)$ and simulates their coupling to print the readout),
never on particular strings; see Appendix~\ref{app:3tape}.

\subsection{The Coupled World-Regulator System}
We work with 3-tape Turing machines $W$  and  $R$   (see Figure~\ref{fig:setup} and Appendix~\ref{app:3tape}). We  identify each machine with its minimal self–delimiting program (\(|W|=K(W)\), \(|R|=K(R)\)) \cite{LiVitanyi4th}. 
 A horizon \(N\in\mathbb{N}\) is fixed and all complexities are conditioned on \(N\)  unless otherwise stated.
 \(W\) and  \(R\) interact causally for \(N\) steps, producing a deterministic readout \(O^{(N)}_{W,R}\in\{0,1\}^N\). 
 The dynamical equations are  
\begin{equation}
    O_W = W(O_R), \: O_R = R(O_W).
    \end{equation}

The performance of the regulator is evaluated from the complexity of the output, $K(x)$. Intuitively, a good regulator produces outputs of lower complexity than the unregulated case.
Since \(x=O^{(N)}_{W,R}\) is computable from \((W,R,N)\),
\begin{equation}\label{eq:detbound}
K( x) \;\le\; K(W,R) + O(1) \;=\; K(W) + K(R) - M(W\!:\!R) + O(1).
\end{equation}

 \begin{figure} [t!]
    \centering
    \includegraphics[width=0.99\linewidth]{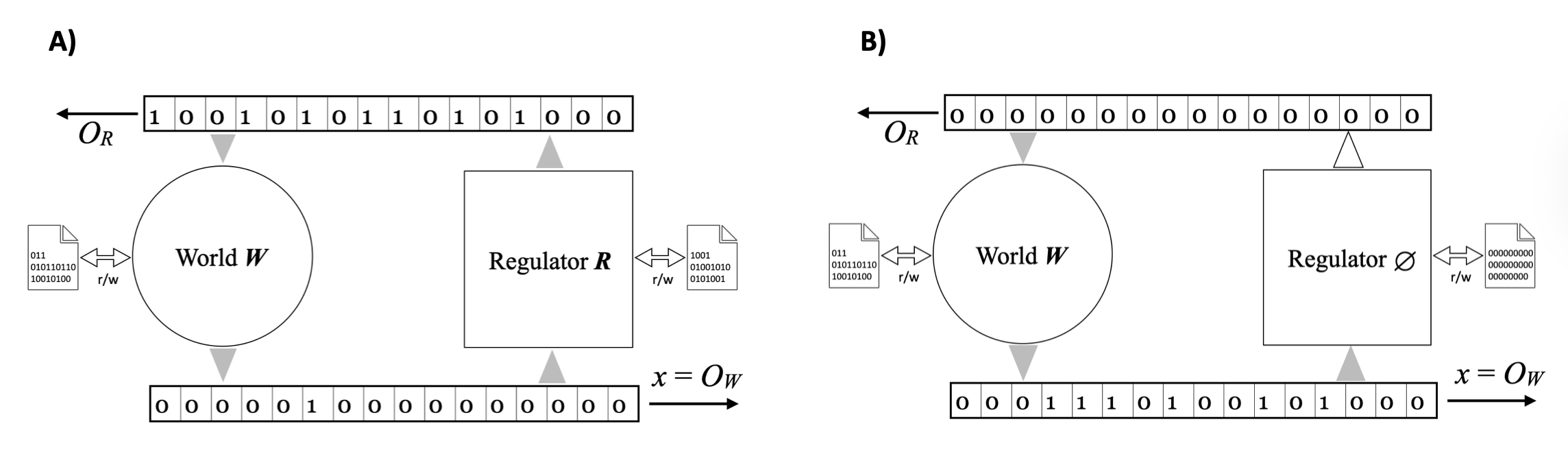}
    \caption{Regulation scenario. A) A good regulator $R$ interacts with the world $W$ so that the readout $x=O_W$ of the world’s output is clamped to a simple, highly compressible sequence (e.g., almost all zeros).  B) When the regulator is turned off, the output is more complex.
    }
     \label{fig:setup}
\end{figure}

To disentangle the role of \(R\) from the coarse event “\(K(O^{(N)})\) is small,” we fix a \emph{null} regulator \(\varnothing\) (where \(R\)'s output is set to zero). We compare the events
\begin{equation}
E^R_a:\ K(O^{(N)}_{W,R} )= a
\quad\text{vs}\quad
E^\varnothing_b:\ K(O^{(N)}_{W,\varnothing} )= b,
\end{equation}
with \(b>a\). Event \(E^\varnothing_b\) rules out worlds that produce a simple output \emph{without} regulation; the intersection \(E^R_a\wedge E^\varnothing_b\) isolates \(R\)’s contribution.

For notational simplicity,  the on–case and off–case readouts are also expressed as
\[
x:=O^{(N)}_{W,R},\qquad y:=O^{(N)}_{W,\varnothing}.
\]

For a fixed time horizon, we write $O_W$ for the full output produced by $W$ when coupled to $R$.



In the next sections, we provide our main results regarding mutual information between world and regulator, and implications for inferring agent-like behavior in the regulator.  

\section{Probabilistic Regulator Theorems}
\label{sec:ugar-final}

\subsection{Posterior form, given the observed $x$}
 
\begin{lemma}[Program posterior given \(x\)]
\label{thm:posterior-x-ugar}
With prefix prior \(P(p)=2^{-|p|}\) and deterministic likelihood \(P(x\mid p)=\mathbf{1}\{U(p)=x\}\),
\[
P(p\mid x)=\frac{2^{-|p|}}{m(x)}.
\]
Consequently, by \eqref{eq:coding-ugar},
\[
\frac{1}{c_2}\,2^{\,K(x)-|p|}\ \le\ P(p\mid x)\ \le\ \frac{1}{c_1}\,2^{\,K(x)-|p|}.
\]
\end{lemma}

\begin{proof}
For any finite string \(x\),
\[
K(x):=\min\{|p|:U(p)=x\},\qquad
m(x):=\sum_{p:\,U(p)=x}2^{-|p|}.
\] (recall Eq.~\ref{eq:levin}).
The Coding Theorem gives machine-dependent constants \(c_1,c_2>0\) with
\begin{equation}
\label{eq:coding-ugar}
c_1\,2^{-K(x)} \ \le\ m(x) \ \le\ c_2\,2^{-K(x)}.
\end{equation}
 Now, briefly, Bayes’ rule yields \(\Pr\{p\mid x\}=2^{-|p|}/m(x)\); apply \eqref{eq:coding-ugar}.
 In more detail, place the prefix prior $P(p)=2^{-|p|}$ on programs $p$ and use the deterministic likelihood
$P(x\mid p) = \mathbf{1}\{U(p)=x\}$.
Then the evidence is $P(x)=m(x)$ and the posterior is
\[
  P(p\mid x)
  \;=\;
  \frac{P(x\mid p)P(p)}{P(x)}
  \;=\;
  \begin{cases}
    \displaystyle \frac{2^{-|p|}}{m(x)}, & U(p)=x,\\[6pt]
    0, & \text{otherwise.}
  \end{cases}
\]

Then, 
for any $p$ with $U(p)=x$,
\[
  \frac{1}{c_2}\,2^{K(x)-|p|}
  \;\le\;
  P(p\mid x)
  \;=\;
  \frac{2^{-|p|}}{m(x)}
  \;\le\;
  \frac{1}{c_1}\,2^{K(x)-|p|}.
\]

The relation between $K(x)$ and $m(x)$ holds only up to an additive $O(1)$ term in $K$,
which becomes a multiplicative constant on $m(x)$. This $O(1)$ ambiguity is unavoidable
and depends on the choice of universal prefix machine $U$; $c_1,c_2$ absorb exactly this
machine-dependent slack.
\end{proof}

Now, in our setting the \emph{world} \(W\) and \emph{regulator} \(R\) are programs that interact for \(N\) steps, producing the on-case readout
\(x:=O^{(N)}_{W,R}\).
A fixed, constant-overhead \emph{wrapper} decodes a shortest description of \((W,R)\) and simulates the coupling to print \(x\)
(\emph{decode\,+\,simulate}); if \(p_{W,R}\) denotes this canonical code, then
\begin{equation}
\label{eq:canon-ugar}
|p_{W,R}|=K(W,R)+O(1),
\qquad
P\big((W,R)\mid x\big)\ \in\ \Big[\tfrac{1}{\tilde c_2},\tfrac{1}{\tilde c_1}\Big]\cdot 2^{\,K(x)-K(W,R)},
\end{equation}
for constants \(\tilde c_i:=2^{O(1)}c_i\).

Now we can use the definition of  \textit{mutual algorithmic information} (up to the usual $O(\log)$ slack) to write
\[
M(W{:}R)\ =\ \K(W)+\K(R)-\K(W,R)\  
\]
and derive our first result:
\begin{theorem} \label{thm:first}
\begin{equation}
\label{eq:first}
P\big((W,R)\mid x\big)\ \in\ \Big[\tfrac{1}{\tilde c_2},\tfrac{1}{\tilde c_1}\Big]\cdot 2^{\,K(x)-K(W)-K(R) +M(W{:}R)} < \frac{1}{\tilde c}\,  2^{M(W{:}R)}
\end{equation}
\end{theorem}

\subsection{The Good  Algorithmic Regulator and Posterior with Contrast}
For our second result, we first define the \textit{Good Algorithmic Regulator (GAR)}. 
\begin{definition}[Good Algorithmic Regulator, contrastive]
\label{def:uGAR}
Given the on/off complexities and gap
\[
a:=K(O^{(N)}_{W,R}),\qquad b:=K(O^{(N)}_{W,\varnothing}),\qquad \Delta:=b-a.
\]
we say that \(R\) is a \emph{good algorithmic regulator of gap \(\Delta\)} for \(W\) at horizon \(N\)
if \(\Delta>0\).
\end{definition}

\begin{lemma}[OFF run lower-bounds the world]
\label{lem:off-ugar}
There exists \(c_0=O(1)\) such that
\[
K\!\big(O^{(N)}_{W,\varnothing}\big)\ \le\ K(W)+c_0\qquad\Rightarrow\qquad K(W)\ \ge\ b-c_0.
\]
\end{lemma}

\begin{proof}
Given \((W,\varnothing,N)\), the wrapper simulates the OFF dynamics and prints \(O^{(N)}_{W,\varnothing}\) with \(O(1)\) overhead.
\end{proof}

With this definition we can now state and prove our main theorem. 
\begin{theorem}[Probabilistic regulator theorem]
\label{thm:ugar}
Let \(O^{(N)}_{W,R}\) and $\mathsf{E}_b^R$ be observed and let \(\Delta:=K(O^{(N)}_{W,\varnothing})-K(O^{(N)}_{W,R})\).
Then there exists \(C>0\) such that
\[
P\!\big((W,R)\mid O^{(N)}_{W,R},\mathsf{E}_b^R\big)\ \le\ C\;\cdot\;2^{\,M(W{:}R)}\;2^{-\Delta}.
\]
Equivalently, every bit by which \(M(W{:}R)\) falls short of \(\Delta\) costs a factor \(\approx 2^{-1}\) in posterior support.
\end{theorem}

\begin{proof}[Proof]
\textit{(i) Posterior via wrapper.} From Eq.~\eqref{eq:canon-ugar},
\(
\log_2 P((W,R)\mid x)\ \le\ K(x)-K(W,R)+O(1)=a-K(W,R)+O(1).
\)

\textit{(ii) Decompose \(K(W,R)\).} 
We use the exact mutual information
$
M(W{:}R):=K(W)+K(R)-K(W,R)
$,
\(
K(W,R)=K(W)+K(R)-M(W{:}R),
\)
hence
\[
K(x)-K(W,R)=a-K(W)-K(R)+M(W{:}R).
\]

\textit{(iii) Insert OFF bound (where \(b\) enters).} By Lemma~\ref{lem:off-ugar}, \(K(W)\ge b-c_0\), so
\[
K(x)-K(W,R)\ \le\ M(W{:}R) - (b-a) - K(R) + c_0\ =\ M(W{:}R) - \Delta - K(R) + c_0.
\]

\textit{(iv) Exponentiate and absorb constants.} Exponentiating and using \(2^{-K(R)}\le 1\) gives
\(
P((W,R)\mid x,\mathsf{E}_b^R)\ \le\ C\, 2^{\,M(W{:}R)}\,2^{-\Delta}
\)
for a constant \(C\) absorbing \(2^{c_0}\) and the wrapper Coding-Theorem constants.
\end{proof}

\paragraph{Clarifications.}
(i) \emph{Where does \(b\) appear?} Only via Lemma~\ref{lem:off-ugar}, which says the OFF run lower-bounds \(K(W)\).
We never need to compute \(b\) explicitly.  
(ii) \emph{Why can we drop \(2^{-K(R)}\)?} A slightly sharper bound is
\(P((W,R)\mid x,\mathsf{E}_b^R)\le C\,2^{M(W{:}R)}2^{-\Delta}2^{-K(R)}\). Since \(K(R)\ge 0\), dropping \(2^{-K(R)}\le 1\)
keeps the focus on the two interpretable scalars \(M\) and \(\Delta\) without changing the exponential scaling.  
(iii) \emph{Architecture-agnostic.} The proof only uses the computable wrapper \((W,R,N)\mapsto x\).
Whether \(R\) is open- or closed-loop does not affect the posterior algebra. iv) The posterior on the left of Theorem~\ref{thm:ugar} is conditioned on the
\emph{on‑case} observation $x$ only. The \emph{off‑case} run is used solely to
supply a numeric lower bound $b:=K(O^{(N)}_{W,\varnothing})$, which implies
$K(W)\ge b-O(1)$ by simulation. Formally, we phrase the result as a bound
on $\Pr((W,R)\mid x,\mathsf{E}_b^R)$, where $\mathsf{E}_b^R$ is the side‑event
``$K(O^{(N)}_{W,\varnothing})= b$''. 

As a consequence of Theorem~\ref{thm:ugar}, one can bound individual posterior masses by $O(2^{K(x)-K(W,R)})$. This implies an exponential tail: $\Pr{M(W:R) \le \Delta - k} = O(2^{-k})$. In other words, $M(W:R)$ is concentrated within $O(1)$ of its maximum $\Delta$. I.e., there exists $C'>0$ (machine/wrapper dependent only) such that for all integers $k\ge 0$,
\[
\Pr\!\Big\{\, M(W{:}R)\ \le\ \Delta - k\ \Bigm|\ x,\ \mathsf{E}_b^R \Big\}\ \le\ C'\,2^{-k}.
\]
 
\begin{plainbox}
\textbf{How to read (and use) Theorem~\ref{thm:ugar}.}
\begin{enumerate}[leftmargin=1.2em,itemsep=2pt,topsep=2pt]
\item \emph{What we measure:} compute the on/off complexities \(a=K(O^{(N)}_{W,R})\) and \(b=K(O^{(N)}_{W,\varnothing})\) (in practice: fixed MDL code lengths); their difference \(\Delta=b-a\) is the \emph{compressibility advantage}.
\item \emph{What the bound says:} for any explanation \((W,R)\) of the observed \(x\), the universal posterior weight is penalized as \(2^{-\Delta}\) unless the pair shares structure: larger \(M(W{:}R)\) \emph{compensates} the penalty.
\item \emph{Practical rule of thumb:} sustained large \(\Delta\) across tasks makes low \(M(W{:}R)\) exponentially unlikely. If off‑case \(b\) is already small, \(\Delta\) will be small—choose a diagnostic readout so the null is not trivially simple.
\end{enumerate}
\end{plainbox}
\subsection{Inferring the Objective Function and Planner (As-If Agent)}
\label{sec:objective-ait}
We next provide a simple theorem regarding the role of complexity as an objective function.
\begin{theorem}[On/Off evidence equals \emph{unconditioned} complexity gap]
\label{thm:onoff-BF-uncond}
%
Under the universal a priori semimeasure,
\begin{equation}
\label{eq:uncond-bf}
\log_2\frac{m(O^{(N)}_{W,R})}{m(O^{(N)}_{W,\varnothing})}
\;=\;
K(O^{(N)}_{W,\varnothing})\;-\;K(O^{(N)}_{W,R})\ \pm O(1).
\end{equation}
Equivalently, writing the  on/off gap as
\(
\Delta:=K(O^{(N)}_{W,\varnothing})-K(O^{(N)}_{W,R}),
\)
we have
\(
m(O^{(N)}_{W,R})/m(O^{(N)}_{W,\varnothing})\;=\; \Theta\!\big(2^{\,\Delta}\big).
\)
Hence, on the realized pair \((O^{(N)}_{W,R},O^{(N)}_{W,\varnothing})\), maximizing the likelihood of “ON over OFF’’ is \emph{equivalent} (up to a constant factor) to minimizing \(K(O^{(N)}_{W,R})\) or, equivalently, maximizing the gap \(\Delta\).
\end{theorem}

\begin{proof}
By the Coding Theorem there exist machine-dependent constants \(c_1,c_2>0\) such that
\(c_1 2^{-K(z)}\le m(z)\le c_2 2^{-K(z)}\) for any string \(z\).
Apply this to \(x\) and \(O^{(N)}_{W,\varnothing}\), take base-2 logs, and subtract:
\[
- \log_2 m(O^{(N)}_{W,R}) \;=\; K(O^{(N)}_{W,R}) \pm O(1),\qquad
- \log_2 m(O^{(N)}_{W,\varnothing}) \;=\; K(O^{(N)}_{W,\varnothing}) \pm O(1),
\]
so
\(
\log_2 \tfrac{m(O^{(N)}_{W,R})}{m(O^{(N)}_{W,\varnothing})} = K(y)-K(O^{(N)}_{W,R}) \pm O(1).
\)
\end{proof}

This  statement compares \emph{two different strings} (the realized ON and OFF outputs) and aligns with the contrastive quantities used elsewhere. 
The log universal Bayes factor for ``ON vs.\ OFF'' is seen to equal the  complexity gap
\(
\Delta\pm O(1).
\)
Thus, on each episode, a regulator behaves \emph{as if} it were maximizing the scalar
$
\Delta ,
$
equivalently minimizing $K\!\big(O^{(N)}_{W,R}\big)$.

 Thus, given a regulator $R$ that persistently reduces the readout’s complexity relative to a null baseline $\varnothing$ (the GAR setting of Def.~\ref{def:uGAR}), we can justify---on purely observational grounds---that $R$ behaves \emph{as if it were minimizing a scalar objective}. The objective should be canonical (not post hoc) and usable across episodes/tasks.

\section{Discussion}
\label{sec:discussion}

We can summarize now our results: 
\paragraph{First regulator result: posterior form, given the observed $x$ (Th.~\ref{thm:first}).}
 By Solomonoff induction and the Coding Theorem \cite{Solomonoff64a,solomonoffFormalTheoryInductive1964a,ZvonkinLevin70,Vitanyi2013,Hutter2007}, we showed that
\begin{equation}
\label{eq:intro-core}
\Pr\!\big((W,R)\mid x\big)\;=\;\frac{2^{-\K(W,R)+O(1)}}{m(x)}
\sim  2^{\,\K(x)-\K(W,R)} < \frac{1}{\tilde c}\,  2^{M(W{:}R)}
\end{equation}
Thus \emph{shorter joint generators are exponentially preferred}; every extra bit in $\K(W,R)$ halves the posterior weight. Decomposing
\begin{equation}
    \K(W,R)\;=\;\K(W)+\K(R)\;-\;M(W{:}R)\ \pm O(\log)
\end{equation}
 shows that, \emph{at fixed marginals} $\K(W),\K(R)$, the posterior is \emph{exponentially tilted} in the algorithmic mutual information $M(W{:}R)$:
each extra bit of $M(W{:}R)$ multiplies posterior odds by $\approx 2$.

\paragraph{Second regulator result: posterior with contrast (Th.~\ref{thm:ugar}).}
\emph{Without contrast}, the story is pure Occam: \eqref{eq:intro-core} anchors the posterior near $\K(W,R)\approx \K(x)$ with a geometric excess-length tail; for fixed $\K(W), \K(R)$, this yields a high-probability lower bound on $M(W{:}R)$ roughly $\K(W)+\K(R)-\K(x)$.  
\emph{With contrast}, if turning the regulator \emph{on} yields $\K(O^{(N)}_{W,R} )=  a$ while the \emph{off} case has $\K(O^{(N)}_{W,\varnothing} )= b$ with $b>a$, then any explaining $(W,R)$ obeys
\[
\Pr((W,R)\mid x)\ \le\ C\,2^{\,M(W{:}R)}\,2^{-\Delta}\!,
\]
so \emph{low mutual information is exponentially disfavored} as the gap $\Delta =b-a$ grows. In both regimes, the operational slogan holds: \emph{see a simple string ($\K(x)$ small), suspect a simple generator ($\K(W,R)$ small)}, and at fixed marginals this means \emph{suspect larger $M(W{:}R)$}.

The inutition behind these results is that
\emph{seeing a simple string suggests its generation by a simple program.}
Formally, for the coupled hypothesis \(P=(W,R)\) (wrapped as a single self‑delimiting program), observing \(x=O_W^{(N)}\) yields the Solomonoff posterior
$
\Pr(P\mid x)\; 
\sim 2^{\,\K(x)-\K(P)},
$ 
by the Coding Theorem \cite{Solomonoff64a,solomonoffFormalTheoryInductive1964a,ZvonkinLevin70,Vitanyi2013,Hutter2007}.
Every extra bit of joint description \(\K(P)=\K(W,R)\) halves posterior weight.
This is the quantitative Occam tilt that operationalizes the slogan above.
 

The posterior mass of joint programs longer than \(\K(x)+k\) decays geometrically:
\[
\Pr\{\,\K(W,R)\ge \K(x)+k\mid x\,\}\ \le\ 2C\,2^{-k}.
\]
Hence the typical joint length is near \(\K(x)\). If \(\K(W)\) and \(\K(R)\) are externally constrained (e.g., by design or prior knowledge), this tail translates directly into a \emph{lower} posterior bound on \(M(W{:}R)\) of the form
\(
M(W{:}R)\gtrsim \K(W)+\K(R)-\K(x)-O(\log(1/\delta))
\)
with posterior confidence \(1-\delta\).

Our results are most informative when the observed readout $O_W^{(N)}$ is \emph{simple}. If $K(O_W^{(N)})$ is large, the posterior constraints on joint complexity and on mutual information are inherently weak.
From the geometric tail, for any $\delta\in(0,1)$ there exists $k=\lceil\log_2(2C/\delta)\rceil$ such that, with posterior probability at least $1-\delta$,
\[
K(W,R)\ \le\ K(O_W^{(N)})+k.
\]
At fixed marginals $K(W)$ and $K(R)$ this yields
\[
M(W{:}R)\ \ge\ K(W)+K(R)-K(O_W^{(N)})-k-O(\log)\qquad\text{with probability }\ge 1-\delta.
\]
Hence, if $K(O_W^{(N)})$ is \emph{large} (comparable to $K(W)+K(R)$), the lower bound on $M(W{:}R)$ may be trivial (near $0$ up to logs). Intuitively, a complex output does not force shared structure. It is compatible with a complex joint generator even when $W$ and $R$ share little algorithmic information.

On the other hand,  the strength of the conclusion depends on the \emph{gap} $\Delta =b-a$:
\[
\Pr\!\big((W,R)\mid O_W^{(N)},\mathsf{E}_b^R\big)\ \le\ C\,2^{\,M(W{:}R)}\,2^{-\Delta},\qquad
\Pr\!\big\{M(W{:}R)\le \Delta-k\ \big|\ O_W^{(N)},\mathsf{E}_b^R\big\}\ \le\ C'2^{-k}.
\]
Thus even if $a=K(O^{(N)}_{W,R} )$ is not very small, a \emph{large} off/on gap still enforces a \emph{large} posterior $M(W{:}R)$. In other words, contrast rescues identifiability of shared structure: the evidence scales exponentially in $\Delta$.

In the same universal calculus, regulation carries a canonical scalar interpretation:
\emph{runtime} behavior is as if minimizing $K(O_W^{(N)} )$ (i.e., maximizing the on/off gap $\Delta$),
and \emph{design-time} comparison across explanations favors larger $M(W{:}R)-\Delta$ via the GAR
posterior tilt. This supplies an MDL/Occam objective grounded in the coding theorem (not an
ad hoc utility) and complements the IMP’s structural requirements.

 We note that a low \(\K(O_W^{(N)})\) \emph{alone} does not prove high \(M(W{:}R)\); it concentrates posterior mass on \emph{short} joint generators \(P\). High \(M(W{:}R)\) follows (i) when \(\K(W)\) and \(\K(R)\) are fixed/known, or (ii) when contrast pins \(\K(W)\) high via the off case. Without such constraints, short \(P\) could also arise from individually simple \(W\) and \(R\).

\paragraph{Third regulator result: as-if Objective-function minimization (Th~\ref{thm:onoff-BF-uncond}).}
On the realized $O_W^{(N)}$, the conditional Coding Theorem gives
$\log_2\!\big(m(O_W^{(N)})/m(O^{(N)}_{W,\varnothing} )\big)=
K(O^{(N)}_{W,\varnothing} )-K(O_W^{(N)})$.
Thus, the \emph{runtime} scalar  to minimize  is $K(O_W^{(N)})$. Together with the above, this implies that the regulator is acting (as-if) like an algorithmic agent (with a model of the world, objective function and planner).

Theorem~\ref{eq:uncond-bf} is a \textit{representation} statement--- not a mechanism: $R$ need not compute $K$, but persistent large $\Delta$ is exactly what maximizes universal evidence for ``ON'', and  it simultaneously makes low $M(W{:}R)$ exponentially unlikely.
 For a \emph{mechanistic} objective beyond the Minimum Description Length (MDL) evidence, three constructive routes are standard and complementary. First, in \emph{Linear Time-Invariant (LTI)} plants the \emph{Internal Model Principle} makes a structural claim—perfect robust regulation for a specified signal class requires embedding a dynamical copy of the exosystem in the controller—and optimal stabilizing designs arise from explicit quadratic/convex costs (e.g., the \emph{Linear Quadratic Regulator, LQR}); in the nonlinear case, output‑regulation theory yields constructive regulators under solvable regulator equations together with immersion/detectability and (local) zero‑dynamics stability \cite{FrancisWonham1975AMO,FrancisWonham1976Automatica,Sontag2003,Isidori45168,huangNonlinearOutputRegulation2004,priscoliAdaptiveObserversNonlinear2006,AndersonMoore1990}. Second, in \emph{inverse optimal control} and \emph{inverse reinforcement learning (IRL)}, trajectories that satisfy \emph{Karush--Kuhn--Tucker (KKT)} regularity allow identification of a cost $J$ (up to equivalences) whose minimizers reproduce the behavior; in discrete settings, IRL recovers reward functions consistent with observed policies \cite{NgRussell2000,AbbeelNg2004,Ziebart2008}. Third, in \emph{revealed‑preference} analysis, if cross‑episode choices satisfy the \emph{Generalized Axiom of Revealed Preference (GARP)}, Afriat and Varian guarantee the existence of a strictly increasing, concave utility that rationalizes the data, while Debreu’s representation and the Savage/Karni--Schmeidler frameworks provide (state‑dependent) expected‑utility forms under their axioms \cite{Afriat1967,Varian1982,Debreu1954,Savage1954,KarniSchmeidler2016}. 

\paragraph{Planner/policy representation (as‑if agent).}
Any deterministic causal regulator $R$ induces a computable \emph{policy}
$\pi_R:\mathcal{H}_t\!\to\!\mathcal{A}$ mapping the coupled history
$h_t$ (past interface I/O up to time $t$) to the next actuator symbol.
This is simply the operational semantics of $R$ viewed as a function of histories.

The coding‑theorem Bayes‑factor identity (Thm.~\ref{thm:onoff-BF-uncond}) 
supplies a canonical scalar 
such that, \emph{on the realized episode}, the sequence of actions produced by $\pi_R$
is \emph{as if} chosen to maximize $J$ subject to the world dynamics. Together with the
algorithmic “internal model’’ conclusion $M(W{:}R)>0$ (i.e., $K(W\mid R)<K(W)$), this yields
the standard agent triad:
\[
\text{(model) } M(W{:}R)>0,\qquad
\text{(objective) } J(x)=K(y)-K(x),\qquad
\text{(policy/planner) } \pi_R.
\]
\emph{Interpretation.} This is a representation statement, not a claim that $R$ explicitly
solves an optimization problem or contains a modular planner. The existence of $\pi_R$ is tautological
for any deterministic $R$; the “as‑if’’ objective follows from the universal evidence identity above.
Across \emph{tasks/episodes}, if the induced choices satisfy standard consistency axioms (e.g., GARP),
classical revealed‑preference theorems guarantee the existence of a (monotone, concave) utility that
rationalizes the behavior \cite{Afriat1967,Varian1982}; and in dynamical settings, inverse
optimal control / inverse RL constructs a cost for which the observed policy is (near‑)optimal \cite{NgRussell2000,AbbeelNg2004}.
Thus, given (i) algorithmic model content $M(W{:}R)>0$ and (ii) the canonical scalar $J$ from the
coding‑theorem calculus, interpreting the regulator as carrying a \emph{policy/planner} is both natural
and technically justified.

\subsection*{Why AIT is needed}
Our results are \emph{single-episode} and \emph{distribution-free}: they make statements about an individual realized readout $x$ and about the pair $(W,R)$ as concrete programs, without positing a stochastic source. Classical (Shannon) information theory quantifies \emph{expected} code lengths and mutual information with respect to a specified probability law; entropy $H(X)$ and mutual information $I(X;Y)$ are undefined without a distribution, and asymptotic statements (AEP/typical sets) further require ergodicity/mixing assumptions  \cite{CoverThomas2006}.  In our setting, there is no given probabilistic model over worlds, regulators, or outputs—indeed, the point is to \emph{infer} model content from a single realized $x$.

AIT supplies exactly the missing calculus. First, it provides a canonical, machine‑invariant complexity for \emph{individual} strings, $K(x)$, and a universal a priori \emph{semi}measure $m(x)$ (Solomonoff–Levin), connected by the Coding Theorem: $-\log m(x)=K(x)\pm O(1)$ \cite{Solomonoff64a,solomonoffFormalTheoryInductive1964a,zvonkinCOMPLEXITYFINITEOBJECTS1970}. This yields a universal Occam posterior over programs, 
$
\Pr\!\big(p\mid x\big)  \asymp\ 2^{\,K(x)-|p|},
$
from which (i) the geometric excess‑length tail and (ii) our \emph{contrastive} tilt bounds follow. No analogue exists in Shannon’s framework without positing an external prior over programs; there is no “canonical” $\Pr(p)$ or $\Pr(x)$ in Shannon theory.

Second, AIT lets us formalize “the regulator contains a model of the world” as \emph{algorithmic dependence}, i.e.\ positive mutual algorithmic information $M(W{:}R)>0$ (equivalently $K(W\mid R)<K(W)$), a notion defined for individual objects and invariant up to $O(1)$ (\cite{LiVitanyi4th}). By contrast, Shannon’s $I(W;R)$ requires a joint distribution over $(W,R)$, which is neither given nor natural here.

Third, our key inequalities explicitly use $m(\cdot)$ and prefix complexity: the posterior tilt $2^{\,K(x)-K(W,R)}$, the OFF‑run lower bound on $K(W)$ by simulation, and the contrastive penalty $2^{-\Delta}$ all rely on the Coding Theorem and Kraft–McMillan properties of \emph{prefix} programs—again, objects absent from Shannon’s ensemble‑level calculus.

Finally, while one can approximate $K(\cdot)$ with MDL/codelengths in practice, MDL’s justification itself rests on the AIT view that \emph{shorter descriptions are better} and on the coding‑theorem linkage between description length and (universal) probability  \cite{Grunwald2007}. In short: AIT provides the universal prior ($m$), object‑level complexities ($K$), and mutual algorithmic information ($M$) needed to turn the informal slogan “see a simple string, suspect a simple generator” into posterior and contrastive theorems—none of which can be stated in Shannon’s framework without ad hoc model classes and priors.

\subsubsection*{Relation to the Internal Model Principle (IMP)}

In the IMP, the closed loop is $(E,C,P)$: an \emph{autonomous} exosystem $E$ (no inputs and no explicit time dependence, e.g.\ $\dot w = S w$), a controller $C$ (the regulator), and a plant $P$. The regulated error is $e = r - y$, where the reference $r$ and disturbances are generated by $E$ and $y$ is measured from $P$ \cite{FrancisWonham1975AMO,FrancisWonham1976Automatica}. In our notation, we group the \emph{World} as $W=(E,P)$ and take the \emph{Regulator} as $R \equiv C$ (see Figure~\ref{fig:wr-imp-yourlayout} and Table~\ref{tab:thermostat-mapping} for the comparison of the two frameworks in the case of a thermostat).

The assumptions in IMP theorems are:
(i) Classical necessity is sharpest for \emph{finite-dimensional LTI} plants (linear, time-invariant) with exogenous signals generated by a finite-dimensional, neutrally stable LTI $E$; stabilizability/detectability and robustness (one fixed $C$ works for a plant neighborhood) are standard \cite{FrancisWonham1975AMO,FrancisWonham1976Automatica}. 
(ii) The structural conclusion is \emph{internal-model necessity}: perfect robust regulation for the specified signal class requires that $C$ embed a dynamical copy of $E$ (e.g., integrators for steps, oscillators for sinusoids); in MIMO, a $p$-copy is needed. 
(iii) Nonlinear generalizations (output regulation) require solvability of the \emph{regulator equations}, suitable \emph{immersion/detectability}, and (local) stability of the \emph{zero dynamics}; guarantees are typically \emph{local/semiglobal}, and necessity is not universal \cite{Isidori45168,huangNonlinearOutputRegulation2004,priscoliAdaptiveObserversNonlinear2006}. 
(iv) Infinite-dimensional/distributed settings and periodic signals may require infinite-dimensional internal models; technicalities arise with unbounded I/O operators \cite{BinEtAl2022-AR}.

In the AIT formulation (here), we assume:
(i) \emph{Architecture-agnostic}: no required split into $E$ vs.\ $P$, and no specified place where $R$ enters the causal path; we only assume a computable wrapper mapping $(W,R,N)\mapsto O_W$ for a fixed horizon $N$. 
(ii) \emph{Deterministic, closed} coupling of world and regulator (no stochastic noise sources into $W$); statements are \emph{distribution-free} and about the realized sequence. 
(iii) ``Model'' means \emph{algorithmic dependence}: $M(W\!:\!R)>0$ (equivalently $K(W|R)<K(W)$), not a literal dynamical replica. 
(iv) The main necessity is \emph{probabilistic}: a positive on/off complexity gap $\Delta = K(O_{W,\emptyset}) - K(O_{W,R})$ exponentially tilts the universal posterior against explanations with small $M(W\!:\!R)$; no linearity, smoothness, or regulator-equation conditions are imposed. See Secs.~2--6 of this work. 

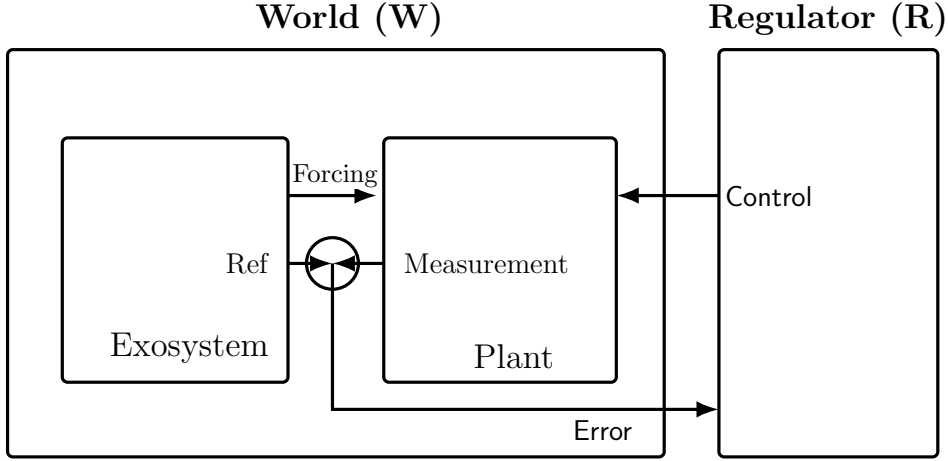
\begin{figure}[t]
\centering

\newcommand{\figscale}{0.90} 

\begin{tikzpicture}[
  scale=\figscale,
  transform shape, 
  font=\sffamily,
  >=Latex,
  box/.style={draw,line width=1.2pt,rounded corners=2pt},
  wire/.style={line width=1.2pt},
  lbl/.style={fill=white,inner sep=1pt}
]

\draw[box] (0,0) rectangle (9.6,6);
\node[font=\bfseries\large] at (5,6.45) {World (W)};

\draw[box] (10.4,0) rectangle (13.6,6);
\node[font=\bfseries\large] at (12.0,6.45) {Regulator (R)};

\draw[box] (0.8,1.1) rectangle (4.1,4.7);
\node[font=\large,anchor=south west] at (1.35,1.25) {Exosystem};

\draw[box] (5.5,1.1) rectangle (8.9,4.7);
\node[font=\large] at (7.4,1.49) {Plant};
\node[font=\normalfont] at (7,2.85) {Measurement};

\coordinate (Sum) at (4.75,2.85);
\draw[line width=1.2pt] (Sum) circle (0.38);

\draw[wire,->] (4.1,3.85) -- (5.4,3.85);
\node[lbl,font=\small] at (4.79,4.15) {Forcing};

\draw[wire,->] (4.1,2.85) -- (Sum);
\node[lbl,anchor=west,font=\normalfont] at (3.15,2.85) {Ref};

\draw[wire,->] (5.5,2.85) -- (Sum);

\draw[wire,->] (Sum) |- (9.9,0.7) -- (10.4,0.7);
\node[lbl] at (8.7,0.4) {Error};

\draw[wire,<-] (8.9,3.85) -- (10.4,3.85);
\node[lbl,anchor=west] at (10.45,3.85) {Control};

\end{tikzpicture}

\caption{To connect the IMP and the AIT formulation used here, we view the World $W$ as a box containing $E$ and $P$; the Regulator/Controller $R$ (or $C$) is a separate box. Arrows depict \emph{Forcing} ($E\!\to\!P$), \emph{Ref} ($E\!\to$ sum), the \emph{Error} path (sum $\downarrow$ to the world boundary and $\rightarrow R$), and \emph{Control} ($R\!\to\!P$).}
\label{fig:wr-imp-yourlayout}
\end{figure}
 
IMP yields a \emph{structural} necessity (internal model in $C$ of $E$) under explicit dynamical hypotheses; the AIT formulation yields an \emph{information‑theoretic} necessity (positive $M(W\!:\!R)$ favored by the data) without assuming linearity, an $E/P$ split, or a particular causal insertion point for $R$. The two are complementary: IMP is the backbone for constructive regulation in structured classes; the AIT view covers unstructured architectures and single episodes with a universal Occam calculus \cite{FrancisWonham1975AMO,FrancisWonham1976Automatica,Sontag2003,Isidori45168,huangNonlinearOutputRegulation2004,BinEtAl2022-AR}.

\paragraph{The home thermostat.}
As an example, consider a home thermostat as regulator/controller.
Let \(P\) be the living room\,{+}\,heater dynamics (thermal capacitance, heat loss, delays) and \(E\) the exogenous processes (setpoint schedule, outdoor weather/solar, occupancy). The \emph{Internal Model Principle} (IMP) states that exact output regulation for a specified signal class is possible only if the controller embeds a copy of the \emph{exosystem} \(E\) that generates those signals (e.g., an integrator for steps, an oscillator for a fixed sinusoid); plant knowledge is used for stabilization/shaping, but the IMP necessity targets \(E\) itself \cite{FrancisWonham1976,ByrnesIsidori1990}. In our \emph{AIT} view, a regulator \(R\) is “good’’ when it makes the realized readout more compressible than a null baseline; a sustained compressibility gap implies that \(R\) shares computable structure with the \emph{whole world} \(W{=}(P,E)\):
\[
M(W{:}R)\;=\;M\big((P,E):R\big)\;=\;M(P{:}R)\;+\;M(E{:}R\mid P)\;\pm O(\log).
\]
A simple on/off thermostat with a deadband tuned to the room time constant typically yields a bounded limit cycle (not zero steady-state error); under IMP it lacks the needed internal model of constants (no embedded integrator), hence it does \emph{not} achieve exact regulation of the “constant’’ class \cite{FrancisWonham1976,AstromMurray2008}. Nevertheless, in the AIT sense it still qualifies as a regulator: its policy encodes a very compressed model spanning \(P\) (heating raises \(T\), room inertia) and weak regularities in \(E\) (quasi-constant setpoint, slowly varying weather), giving \(M(W{:}R)>0\) \cite{ConantAshby1970}. PI/PID or predictive thermostats remedy the IMP shortfall by embedding the appropriate internal model (and often explicit models of \(P\) and aspects of \(E\)) \cite{AstromMurray2008}.

The AIT regulator framework (as well as the original GRT) is  therefore \emph{more general} than IMP: the regulator must carry a model of the world \(W=E\cup P\), where \(P\) is the plant (house/HVAC thermodynamics) and \(E\) the exogenous processes (setpoint schedule, weather/solar/occupancy), and IMP is recovered as a special case when the performance target is exact output regulation over a specified signal class. Under IMP, a controller qualifies for exact regulation only if it embeds a dynamical copy of the \emph{exosystem} that generates the reference/disturbances (e.g., an integrator for steps, oscillators for sinusoids)---no model of \(P\) is required beyond stabilizability/detectability \cite{FrancisWonham1976}; nonlinear output-regulation extends this under additional immersion/detectability and regulator-equation solvability assumptions \cite{ByrnesIsidori1990}. 

\begin{table*}[t]
\centering \scriptsize
\renewcommand{\arraystretch}{1.18}
\setlength{\tabcolsep}{4pt}
\begin{tabularx}{\textwidth}{
>{\raggedright\arraybackslash}p{3.2cm}
>{\raggedright\arraybackslash}X
>{\raggedright\arraybackslash}X
>{\raggedright\arraybackslash}X}
\toprule
\textbf{Role} &
\textbf{IMP language} &
\textbf{AIT language (this work)} &
\textbf{Thermostat instantiation} \\
\midrule
\textbf{Exogenous generator} &
\textbf{Exosystem} $E$: autonomous generator of references/disturbances (no feedback from $C$); exact regulation is defined w.r.t.\ a signal class $\mathcal U$. &
Fold into the \textbf{World} $W$; no architectural split is required (but may still conceptually identify this subpart). &
\textbf{Reference} $r(t)$: setpoint schedule (often clock‑driven). \textbf{Disturbances}: outdoor temperature, solar load, occupancy heat gains. \\
\addlinespace[2pt]
\textbf{Plant} &
\textbf{Plant} $P$: room thermal dynamics + actuator/sensor; used for stabilization/shaping. &
Also inside \textbf{World} $W$. &
$R$–$C$ (thermal) model, heater actuation, heat losses, sensor dynamics/delay. \\
\addlinespace[2pt]
\textbf{Controller / Regulator} &
\textbf{Controller} $C$ (the regulator in IMP). &
\textbf{Regulator} $R$. &
Thermostat logic: bang‑bang with hysteresis, PI/TPI, or scheduled control. \\
\addlinespace[2pt]
\textbf{Measured output} &
$y$. &
World readout $x$ extracted from the transcript (often $x{=}y$ or the error string $e_{1:T}$). &
Indoor temperature $T_{\mathrm{in}}$ (or a weighted error signal). \\
\addlinespace[2pt]
\textbf{Error / objective} &
$e=r-y$; IMP concerns \emph{asymptotic} $e\!\to\!0$ for all $r,d$ in the class $\mathcal U$ (internal model must match $E$). &
Score regulation by \emph{compressibility} of the chosen readout $x$ with $R$ \emph{ON} vs.\ an \emph{OFF} baseline ($R{=}\varnothing$). Define the gap
$\Delta \;=\; K(x_{\text{off}})\;-\;K(x_{\text{on}})$
(practically, use a fixed MDL code $L_C$ in place of $K$). &
Good thermostat $\Rightarrow$ $x_{\text{on}}$ (e.g., temperature or error) stays near a regular deadband pattern $\Rightarrow$ shorter code than the null/open‑loop case (heater OFF or fixed duty). \\
\bottomrule
\end{tabularx}
\caption{Mapping the IMP triple $(E,C,P)$ and the AIT $(W,R)$ view to a simple thermostat. IMP emphasizes an internal model of the exosystem $E$ for \textit{exact} regulation over a signal class; AIT treats $W{=}(P,E)$ jointly and assesses regulation by a compressibility advantage~$\Delta$.}
\label{tab:thermostat-mapping}
\end{table*}

\emph{Our statements are thus complementary and distinct:} in AIT, we work in a distribution‑free, program‑level setting and make no linearity or smoothness assumptions. We remain agnostic about what the regulator needs to model and do not demand exact regulation. We do not assert the existence of a dynamical replica inside $R$. Instead, we show that sustained contrastive compressibility ($\Delta>0$) \emph{tilts} the universal posterior toward pairs $(W,R)$ with larger mutual algorithmic information $M(W{:}R)$, i.e., $R$ carries algorithmic structure about $W$. Thus, “the regulator contains a model” is made precise as $M(W{:}R)>0$ (information‑theoretic dependence), not as an embedded exosystem. The IMP supplies structural necessity for perfect regulation within specified signal classes; our AIT results supply information‑theoretic necessity for observed compressibility advantages, beyond linearity or probabilistic assumptions  \cite{Sontag2003}.

\subsubsection*{Practical estimation of $\K$ and the gap $\Delta$}
\label{subsec:practical}
Our theorems are stated in terms of prefix Kolmogorov complexity, which is not computable.
In practice, one can fix a reference prefix code $C$ and estimate upper bounds, 
\[
\widehat{a}:=L_C\!\big(O^{(N)}_{W,R}\big),\quad
\widehat{b}:=L_C\!\big(O^{(N)}_{W,\varnothing}\big),\quad
\widehat{\Delta}=\widehat{b}-\widehat{a},
\]
with the \emph{same} compressor $C$ used across all conditions. 
Persistent \(\widehat{\Delta}>0\) across tasks is cumulative evidence that explanations with \emph{low}
\(M(W{:}R)\) are exponentially unlikely; maximizing \(\widehat{\Delta}\) is the natural  scalar
objective the regulator appears to optimize on the observed data.

Some standard choices for providing upper bounds to Kolmogorov complexity  are Lempel-Ziv compressors (LZ77/LZ78/LZW).
LZ‑type compressors are universal in a weak sense for stationary ergodic sources and are widely available.
Implementations (gzip, lz4, etc.) are practical proxies for $L_C(\cdot)$ \cite{ZivLempel1977,ruffiniLempelZipComplexityReference2017}.
If both ON and OFF strings are available and a scale‑free sanity check of contrast is needed, we can compute
\[
\mathrm{NCD}(x,y)\ :=\ \frac{C(xy)-\min\{C(x),C(y)\}}{\max\{C(x),C(y)\}},
\]
where $C(\cdot)$ is the chosen code length and $xy$ is concatenation \cite{liSimilarityMetric2004,cilibrasiClusteringCompression2004}. 
NCD is heuristic but can reveal whether $x$ is “closer” to trivial baselines than $y$.

The Block Decomposition Method (BDM) estimates $\K$ by tiling a string (or array) into small blocks whose complexities are looked up from \emph{Coding‑Theorem‑Method} (CTM) tables (exhaustive output frequency statistics of small machines), plus a logarithmic penalty for multiplicities,
$
\widehat{\K}_{\mathrm{BDM}}(x)\ \approx\ \sum_{i}\K_{\mathrm{CTM}}(b_i) + \log m_i,
$
where $b_i$ are distinct blocks and $m_i$ their multiplicities (see \cite{soler-toscanoCalculatingKolmogorovComplexity2014,zenilDecompositionMethodGlobal2016}). This is sensitive to small‑scale algorithmic regularities beyond LZ’s parse statistics; it works on 1D/2D data (but
depends on the chosen CTM table --- size and machine model --- and it suffers from boundary/tiling effects and additive constants that can be sizable for short $N$).

 Finally, alternatives include \emph{learned} compressors based on neural networks. Autoencoder/variational–autoencoder codecs optimize a rate–distortion (thus MDL) objective, with an explicit codelength view via ELBO and practical lossless coding through bits‑back \cite{Grunwald2007,hinton1993autoencoders,kingma2019introvae,townsend2019bitsback,ho2019localbitsback}. In images and video, end‑to‑end trained autoencoders, hyperpriors, and autoregressive priors are now standard \cite{balle2018hyperprior,minnen2018joint,habibian2019video}. Transformer‑based compressors are competitive and rapidly improving—both for images via transformer or hybrid Transformer–CNN codecs \cite{lu2021transformercompression,liu2023mixedtransformercnn} and for general lossless compression using language‑model predictors and specialized transformer compressors \cite{deletang2024lmcompression,valmeekamLLMZipLosslessText2023a}; see also cross‑modal results reported with large models \cite{liLosslessDataCompression2025} and neural codecs for audio \cite{zeghidour2021soundstream}. \emph{From an MDL perspective, these models implement universal codes whose lengths upper‑bound the negative log‑likelihood under the learned generative model}.

To improve discrimination, we can i) use \emph{paired} ON/OFF measurements on the same horizon $N$; report $\widehat{\Delta}$ and its sampling variability across repeats/seeds; ii) 
include trivial controls (e.g.\ all‑zero regulator and randomized regulator) to sanity‑check that $\widehat{\Delta}$ responds in the expected direction; iii)
for finite $N$, complement point estimates with nonparametric tests (paired permutations on $\widehat{\Delta}$ across episodes); iv)
when outputs are multivariate/real‑valued, discretize with a fixed, reported quantization and alphabet before compression.

\section{Conclusion}\label{sec:conclusion}

We developed a contrastive, algorithmic formulation of regulation: a regulator \(R\) is \emph{good} for a world \(W\) at horizon \(N\) when it yields a compressible readout that is strictly more compressible than under a null baseline \(\varnothing\). This places the GRT claim (“good regulators are models”) on an AIT footing.

If switching a regulator on makes a system’s measured output much simpler to describe (i.e., more compressible) than when the regulator is off, then the regulator is very likely to carry non‑trivial information about the world it controls—in the precise Algorithmic Information Theory  sense of positive mutual algorithmic information between world and regulator. The strength of this evidence grows exponentially with the compressibility gap: large $\Delta$ makes explanations with little shared structure vanishingly likely. Practically, this turns the old cybernetics slogan “every good regulator is a model of the system” into a quantitative, testable claim that does not assume linearity, stochastic models, or specific architectures. On each run, the theorem also singles out a canonical scalar objective: the regulator behaves as if it were minimizing the description length of the realized readout (equivalently, maximizing $\Delta$).  

%

Probabilistically, if \(W\) and \(R\) are independently sampled minimal programs (no mutual information), then low readout complexity—and especially the contrastive event “low under \(R\), high under \(\varnothing\)”—is exponentially unlikely in \(|W|\) and \(|R|\). Thus, sustained compressibility relative to baseline is strong evidence that \(R\) shares non‑trivial algorithmic structure with \(W\) (\(M(W\!:\!R)>0\)). This is the AIT face of the Good Regulator idea and complements the Internal Model Principle’s structural necessity results for classical regulation:   the IMP identifies structural necessities for perfect/robust regulation in classical settings, whereas our AIT view applies beyond linearity and probability and turns regulation into a statement about \emph{description length}.  This bridge clarifies in what limited (yet precise) sense the cybernetics aphorism ``good regulators must model'' can be made rigorous \cite{ConantAshby1970,BaezGRT}: successful regulation implies positive mutual algorithmic information between world and regulator.


The result supplies: (i) a distribution‑free, single‑episode diagnostic for “does the controller contain a model?”, (ii) a complement to the IMP (which requires embedding a copy of the signal generator under more restrictive and structured assumptions), and (iii) a simple experimental recipe—fix a lossless compressor, quantize the readout, compute two code lengths (ON vs. OFF), and use their difference $\Delta$ as evidence of model content in the controller.

Finally, the coding‑theorem view identifies a canonical scalar and implicates a planner: runtime minimization of
$K(x)$ (equivalently, maximization of $\Delta$).

All together, these results provide the grounds to justify that if a system is seen to regulate another in the algorithmic sense (reducing the complexity of an output of the regulated system compared to no regulation), we can reasonably infer it is likely that the regulator uses a model of the regulated system and an associated scalar objective function.

\subsection*{Acknowledgments}
The author thanks Francesca Castaldo for discussions and reviewing the manuscript. The author thanks David Wolpert (SFI) for highilighting open questions regarding the  classical regulator theorem.

\bibliographystyle{plainurl}
\bibliography{references}
\clearpage

\begin{landscape}
\begin{table}[p]
\centering
\scriptsize
\renewcommand{\arraystretch}{1.20}
\setlength{\tabcolsep}{5pt}
\begin{tabularx}{\linewidth}{>{\raggedright\arraybackslash}p{3.2cm} X X X}
\toprule
\textbf{Aspect} & \textbf{GRT (Conant--Ashby, 1970)} & \textbf{IMP (Francis--Wonham, 1975/76; Sontag, 2003)} & \textbf{A-GRT (Algorithmic, this work)}\\
\midrule

\textbf{Setting / Objects}
& System $S$, Regulator $R$, Disturbances/Inputs $D$, Outcomes $Z$. Mapping $\psi:(S,R)\!\mapsto\!Z$; compare regulators by entropy of $Z$.
& Plant $P$ in feedback with Controller $C$; exogenous signals from an \emph{exosystem} $E$; regulated output $y$ and error $e=r-y$.
& World $W$ and Regulator $R$ are deterministic causal prefix programs (3‑tape UTM) that interact over interface tapes for horizon $N$; readout $x=O^{(N)}_{W,R}$.\\
\addlinespace

\textbf{Symbols (explicit)}
& $S$ (system), $R$ (regulator), $D$ (disturbance/input), $Z$ (outcome), $H(\cdot)$ (Shannon entropy).
& $P$ (plant), $E$ (exosystem/signal generator), $C$ (controller), $y$ (regulated output), signal class $\mathcal U$ (e.g., steps/sinusoids/polynomials).
& $W$ (shortest world program), $R$ (shortest regulator program), $x:=O^{(N)}_{W,R}$ (ON readout), $y:=O^{(N)}_{W,\emptyset}$ (OFF readout), $K(\cdot)$ (prefix complexity), $M(\cdot\!:\!\cdot)$ (mutual algorithmic information).\\
\addlinespace

\textbf{Definition of ``model''}
& Deterministic mapping/homomorphism $h:S\!\to\!R$ that preserves task‑relevant structure so outcomes have low entropy.
& \emph{Internal model}: a dynamical subsystem embedded in $C$ that reproduces $E$ (controller contains a copy of $E$’s dynamics; in LTI, matching poles such as integrators/resonators).
& \emph{Algorithmic model (program)}: $R$ shares computable structure with $W$—formally $M(W\!:\!R)>0$ (equivalently $K(W\mid R)\!<\!K(W)$); no need for a literal dynamical replica.\\
\addlinespace

\textbf{Notion of ``goodness''}
& “Maximally successful and simple”: minimize $H(Z)$ and avoid unnecessary regulator randomness/complexity.
& \emph{Perfect regulation} for a specified class $\mathcal U$ (exact asymptotic tracking/disturbance rejection, robustness in class).
& \emph{Compressibility of realized readout}: good if $K(x)$ is small at the chosen $N$; use contrastive \emph{gap} $\Delta := K(O^{(N)}_{W,\emptyset}) - K(O^{(N)}_{W,R})>0$.\\
\addlinespace

\textbf{Core Theorem Statement}
& Among regulators that minimize $H(Z)$ and are simplest, there is a deterministic $h:S\!\to\!R$; informally: “every good regulator is (contains) a model of the system.”
& \textbf{Necessity:} perfect regulation for class $\mathcal U$ \emph{requires} $C$ to embed a copy of $E$ (an internal model).
& \textbf{Algorithmic necessity:} with ON $x$ and OFF complexity $K(O^{(N)}_{W,\emptyset})\!=\!b$, the universal posterior obeys
$\Pr((W,R)\mid x,\!E_b^R)\,\le\,C\,2^{\,M(W\!:\!R)}\,2^{-\Delta}$.
Thus sustained $\Delta\!>\!0$ makes low $M(W\!:\!R)$ exponentially unlikely; on the realized episode, maximizing ON over OFF likelihood is equivalent (up to $O(1)$) to minimizing $K(x)$ (i.e., maximizing $\Delta$).\\
\addlinespace

\textbf{Assumptions}
& $Z$ is well‑defined from $(S,R)$ and disturbances; regulators compared by $H(Z)$ and simplicity. \emph{Ref:} \href{https://doi.org/10.1080/00207727008920220}{Conant \& Ashby (1970)}.
& Typically finite‑dimensional LTI; stabilizable/detectable; $E$ autonomous and neutrally stable; exact asymptotic tracking/rejection for $\mathcal U$; robustness in a plant neighborhood. \emph{Refs:} \href{https://doi.org/10.1007/BF01447855}{Francis \& Wonham (1975)}, \href{https://doi.org/10.1016/0005-1098(76)90006-6}{Francis \& Wonham (1976)}, \href{https://doi.org/10.1016/S0167-6911(03)00136-1}{Sontag (2003)}.
& Deterministic closed coupling; fixed universal prefix machine and horizon $N$; $W,R$ are minimal self‑delimiting programs; constant‑overhead wrapper for $(W,R,N)\!\mapsto\!O^{(N)}_{W,R}$; diagnostic readout (contrast usable). In practice, estimate $K(\cdot)$ with fixed MDL codelengths.\\
\addlinespace

\textbf{Restrictions / Limitations}
& “Model” notion is weak (mapping); success tied to entropy of $Z$ (can reward trivial predictable outcomes); no explicit stability claims.
& Sharpest for LTI; nonlinear/output‑regulation extensions add local solvability/detectability/zero‑dynamics stability; necessity generally local/structural.
& Information‑theoretic (not structural) necessity; strength depends on diagnostic $\Delta$; $K(\cdot)$ uncomputable (use fixed compressor/MDL); single‑episode statements (with probabilistic tilt).\\
\addlinespace

\textbf{Scope / Use}
& Conceptual cybernetics link: regulation $\Rightarrow$ representation (\emph{model‑building is compulsory}).
& Design backbone for robust regulation (integral action, embedded oscillators); concrete synthesis constraints.
& Distribution‑free, single‑episode diagnostics; empirical recipe: fix a lossless compressor, quantize readout, compute ON/OFF code lengths, use $\Delta$ as evidence of model content; complements IMP with universal Occam calculus. \emph{AIT refs:} \href{https://link.springer.com/book/10.1007/978-3-030-11298-1}{Li \& Vitányi (2019)}.\\

\bottomrule
\end{tabularx}
\caption{\scriptsize Side‑by‑side comparison of the classical Good Regulator Theorem (GRT), the Internal Model Principle (IMP), and an Algorithmic‑Information‑Theoretic Good Regulator Theorem (A‑GRT). Primary sources (hyperlinked): \href{https://doi.org/10.1080/00207727008920220}{Conant \& Ashby (1970)}, \href{https://doi.org/10.1007/BF01447855}{Francis \& Wonham (1975)}, \href{https://doi.org/10.1016/0005-1098(76)90006-6}{Francis \& Wonham (1976)}, \href{https://doi.org/10.1016/S0167-6911(03)00136-1}{Sontag (2003)}, and \href{https://link.springer.com/book/10.1007/978-3-030-11298-1}{Li \& Vitányi (2019)}.}
\label{tab:grt-imp-agrt}
\end{table}
\end{landscape}

\clearpage

\appendix

\section{Appendix}

\subsection{Setting and core definitions}
\label{app:defs}

\paragraph{Universal machine and prefix complexity.}
Fix a universal \emph{prefix} Turing machine $U$. For any finite binary string $x$,
\[
\K(x):=\min\{|p|:U(p)=x\},\qquad
m(x):=\sum_{p:\,U(p)=x}2^{-|p|}.
\]
By the Coding Theorem there exist machine–dependent $c_1,c_2>0$ with
$c_1 2^{-\K(x)}\le m(x)\le c_2 2^{-\K(x)}$
\cite{solomonoffFormalTheoryInductive1964a,zvonkinCOMPLEXITYFINITEOBJECTS1970, LiVitanyi4th,Vitanyi2013}.

\paragraph{Conditioning convention.}
A finite horizon $N\in\mathbb{N}$ is fixed throughout; unless stated otherwise,
all complexities are implicitly conditioned on $N$, e.g.\ $\K(x):=\K(x\mid N)$ and
$m(x):=m(x\mid N)$.

\paragraph{Machines and transcripts.}
A \emph{world} $W$ and \emph{regulator} $R$ are deterministic causal prefix programs
that interact for $N$ steps via interface tapes. Their closed‑loop interaction produces
a binary readout $x=O^{(N)}_{W,R}\in\{0,1\}^N$. The \emph{off/null} regulator, denoted $\varnothing$,
is the coupling where the regulator’s interface outputs a fixed quiescent symbol (e.g.\ 0) at all steps,
yielding $y=O^{(N)}_{W,\varnothing}$.

\paragraph{Joint description and wrapper.}
A fixed constant‑overhead \emph{wrapper} decodes shortest descriptions of $(W,R)$ and simulates the
coupling to print $O^{(N)}_{W,R}$. Denote by $K(W,R)$ the length of a shortest self‑delimiting
code for the pair. We use standard chain rules (e.g.\ $K(W,R)=K(W)+K(R\mid W)\pm O(1)$).

\paragraph{Mutual algorithmic information.}
For finite strings $x,y$,
\[
M(x{:}y)\ :=\ \K(x)+\K(y)-\K(x,y)\ \pm O(\log(\K(x)+\K(y))).
\]
Equivalently, $M(x{:}y)=\K(x)-\K(x\mid y)\pm O(\log)$ \cite{LiVitanyi4th}.

\paragraph{Good Algorithmic Regulator (contrastive).}
Let $a:=\K(O^{(N)}_{W,R})$ and $b:=\K(O^{(N)}_{W,\varnothing})$. Define the \emph{gap}
\[
\Delta\ :=\ b-a.
\]
We say that $R$ is a \emph{good} algorithmic regulator for $W$ at horizon $N$ if $\Delta>0$.
(In practice, $a$ and $b$ are estimated by fixed MDL codelengths; see \S\ref{subsec:practical}.)

\paragraph{Deterministic upper bound.}
Since the wrapper simulates the coupling, one always has
\[
\K\!\big(O^{(N)}_{W,R}\big)\ \le\ \K(W,R)\ \le\ \K(W)+\K(R)-M(W{:}R)+O(1).
\]

\subsection{Three-Tape Turing Machine} \label{app:3tape}
\begin{definition}[Three-Tape Turing Machine Algorithm]
A \textit{three-tape Turing machine algorithm} is represented by a Turing machine \( T \) with three tapes, and consists of the following components:
\begin{enumerate}
  \item A finite set of states \( Q \), including a designated start state \( q_0 \) and one or more halting states.
  \item A finite alphabet \( \Sigma \), including a blank symbol, used for the input, output, and private tapes.
  \item Three finite tapes, divided into cells, where each cell can contain a symbol from \( \Sigma \). These tapes are designated as the input tape, the output tape, and the non-erase private tape.
  \item A transition function \( \delta: Q \times \Sigma^3 \rightarrow Q \times \Sigma^3 \times \{L, R\}^3 \), defining how the machine moves between states, writes symbols on the three tapes, and moves the tape heads left (L) or right (R) on each tape.
\end{enumerate}
We further identify the state of the private and output states with a set of variables  \(V = \{v_1, v_2, \ldots, v_n\}\), with subsets \(V_{\text{private}}\) and \(V_{\text{output}}\). 
The time evolution of variables in \(V\) is governed by the operation of the Turing machine, as it processes the input, modifies the private tape \(V_{\text{private}}\), and writes to the output tape \(V_{\text{output}}\), according to \(\delta\).
So we can also see an algorithm as a specification of the evolution of a set of variables.

The Turing machine begins in the start state with the input written on the input tape and the other tapes blank. It proceeds according to the transition function, writing into the output and private tapes.  The private tape can be written to but not erased. When the machine reaches a halting state, the output is read from the output tape.

\end{definition}

\subsection{Prefix-free programs vs.\ stop-symbol delimiters (and why it matters)}

\paragraph{Setup.}
Let \(U\) be a universal \emph{prefix} machine: the domain of its halting programs is prefix-free, so no valid program is a prefix of another. The associated (prefix/self-delimiting) Kolmogorov complexity is
\[
K_U(x) \;=\; \min\{\,|p| \;:\; U(p)=x \text{ and } p \text{ is in a prefix-free domain}\,\}.
\]
Working with prefix-free domains aligns program lengths with instantaneous (prefix) codes and invokes Kraft–McMillan inequality, the coding-theoretic backbone that underlies many AIT results, including Levin’s universal distribution and the coding theorem \cite{LiVitanyi4th,Tadaki2010,FortnowNotes}. (See also standard IT references for Kraft–McMillan and prefix codes \cite{XieKraft,CMU_Kraft}.)

\paragraph{Why prefix-freeness is not a mere technicality.}
\begin{enumerate}[leftmargin=1.5em]
\item \textbf{Instantaneous decodability and Kraft sums.}
If the halting programs form a prefix code, then for the multiset of program lengths \(\{\,|p| : U(p)\downarrow\,\}\) we have
\(
\sum_{p} 2^{-|p|} \le 1
\)
by Kraft–McMillan. This lets us interpret \(2^{-|p|}\) as a valid ``budget'' of probability mass per description and leads to semimeasures like Levin’s universal distribution \(m_U(x)=\sum_{U(p)=x}2^{-|p|}\) with \(\sum_x m_U(x)\le 1\). This construction is central to algorithmic probability and to the coding theorem (roughly \(K(x)\approx -\log m(x)\)) \cite{LiVitanyi4th,Sterkenburg2017,ScholarpediaAlgProb,Tadaki2010}. 

\item \textbf{Clean invariance and chaining inequalities.}
The invariance theorem (machine-independence of \(K\) up to \(O(1)\)) and standard chain rules (e.g.\ \(K(x,y)\le K(x)+K(y\mid x)+O(1)\)) are most naturally proved for prefix machines because self-delimitation removes end-of-program ambiguity in compositions and conditional encodings \cite{LiVitanyi4th,FortnowNotes}. 
\end{enumerate}

\paragraph{“Why not just add a stop symbol?”}
Suppose we try to avoid the prefix constraint by allowing programs of the form \(p\#\), where \(\#\) is an end marker.

\begin{itemize}[leftmargin=1.5em]
\item \emph{If the interpreter ignores any trailing bits after \(\#\)}, then any extension \(p\#q\) yields the same computation as \(p\#\). To keep the \emph{domain of halting programs} unambiguous, you must \emph{reject} all extensions \(p\#q\neq p\#\). But rejecting all such extensions is exactly the prefix-free condition in disguise: no valid codeword is a prefix of another. Thus, a well-implemented ``stop-symbol'' machine reduces to a prefix-free machine up to a fixed additive overhead for encoding \(\#\). Consequently, all asymptotic theorems (invariance, coding theorem, bounds using Kraft) remain unchanged up to \(O(1)\) \cite{LiVitanyi4th,FortnowNotes,Tadaki2010}. 

\item \emph{If extensions after \(\#\) are allowed as distinct valid programs}, then the set of halting inputs is \emph{not} prefix-free, Kraft–McMillan can fail, and the sum \(\sum_{U(p)=x}2^{-|p|}\) need not be bounded by 1. This breaks the semimeasure property essential to Levin’s universal distribution and derails the clean link between probability and description length \cite{Sterkenburg2017,ScholarpediaAlgProb}. In short: allowing arbitrary padding after a nominal ``stop'' symbol undermines the probability calculus that AIT relies on. 
\end{itemize}

\paragraph{Implications for our results}
All conclusions in this paper that rely on (i) the coding-theoretic view of programs, (ii) semimeasures like \(m_U\), or (iii) standard chain/invariance bounds continue to hold if one uses a stop-symbol formalism \emph{implemented so that descriptions are self-delimiting in the sense above}. That formalism is equivalent to the prefix-free setting up to \(O(1)\) and thus does \emph{not} change the substance of our arguments or their asymptotic constants. If, however, the stop-symbol scheme admits padded extensions as distinct valid programs, key lemmas using Kraft (and hence bounds derived via \(m_U\) or coding-theorem arguments) may fail or require nonstandard fixes.

\paragraph{Takeaway.}
The ``prefix business'' is not a dispensable technicality; it encodes self-delimitation that makes programs behave like instantaneous codewords. You can implement self-delimitation via explicit markers, but only if you simultaneously forbid any valid extension after the marker—i.e.\ you recover a prefix-free domain. With that in place, none of the conclusions elsewhere in the paper need to change (beyond harmless \(O(1)\) shifts). Without it, several probability/complexity identifications break.

\subsection{Coding Theorems (unconditional and conditional)}
\label{app:coding-theorem}

\paragraph{Setup and notation.}
Fix a universal \emph{prefix} Turing machine $U$.
All logarithms are base~2.
For a finite string $x$, let $K(x)$ be its (prefix) Kolmogorov complexity:
$K(x):=\min\{|p|: U(p)=x\}$.
The \emph{universal a priori semi\-measure} is
\[
m(x)\;:=\;\sum_{p:\,U(p)=x}2^{-|p|}\,.
\]
Since the halting programs of a prefix machine form a prefix code, Kraft–McMillan implies
$\sum_x m(x)\le 1$.

For \emph{conditional} versions, we equip $U$ with a read‑only auxiliary input tape that holds side information $y$.
Define
\[
K(x\mid y)\;:=\;\min\{|p|: U(p,y)=x\}\!,
\qquad
m(x\mid y)\;:=\;\sum_{p:\,U(p,y)=x}2^{-|p|}.
\]
All $O(1)$ terms and constants below depend only on the choice of $U$, never on $x$ or $y$.

\begin{theorem}[Coding Theorem (unconditional)]
\label{thm:coding-uncond}
There exist machine‑dependent constants $c_1,c_2>0$ such that for all finite strings $x$,
\[
c_1\,2^{-K(x)} \;\le\; m(x) \;\le\; c_2\,2^{-K(x)}.
\]
Equivalently,
\[
-\log m(x) \;=\; K(x)\ \pm O(1).
\]
\end{theorem}

\begin{proof}[Proof sketch]
\emph{Lower bound.}
Let $p^\star$ be a shortest program for $x$, so $|p^\star|=K(x)$ and $U(p^\star)=x$.
Then $m(x)\ge 2^{-|p^\star|}=2^{-K(x)}$ (the constant $c_1$ absorbs harmless machine choices).

\emph{Upper bound.}
Because $m(\cdot)$ is a semimeasure, there exists a prefix code with lengths
$\ell(x)\le \lceil-\log m(x)\rceil$ (Shannon–Fano/Kraft–McMillan).
A fixed decoder transforms the codeword for $x$ into $x$, so $K(x)\le \ell(x)+O(1)\le -\log m(x)+O(1)$.
Rearranging gives $m(x)\le c_2\,2^{-K(x)}$.
\end{proof}

\begin{theorem}[Coding Theorem (conditional)]
\label{thm:coding-cond}
There exist machine‑dependent constants $c'_1,c'_2>0$ such that for all finite strings $x,y$,
\[
c'_1\,2^{-K(x\mid y)} \;\le\; m(x\mid y) \;\le\; c'_2\,2^{-K(x\mid y)}.
\]
Equivalently,
\[
-\log m(x\mid y) \;=\; K(x\mid y)\ \pm O(1).
\]
\end{theorem}

\begin{proof}[Proof sketch]
\emph{Lower bound.}
With $p^\star$ a shortest conditional program for $x$ given $y$, we have $U(p^\star,y)=x$, hence
$m(x\mid y)\ge 2^{-|p^\star|}=2^{-K(x\mid y)}$.

\emph{Upper bound.}
For fixed $y$, $m(\cdot\mid y)$ is a semimeasure, so there is a prefix code (depending on $y$) with
$\ell(x\mid y)\le \lceil-\log m(x\mid y)\rceil$ and a fixed decoder (shared across all $y$) that maps codewords
plus $y$ to $x$. Therefore $K(x\mid y)\le -\log m(x\mid y)+O(1)$, which rearranges to the stated upper bound.
\end{proof}

\paragraph{Remarks.}
\begin{itemize}[leftmargin=1.2em,itemsep=2pt]
\item The constants $c_1,c_2,c'_1,c'_2$ (and all $O(1)$ slacks) depend only on the choice of the universal prefix machine $U$; changing $U$ shifts $K(\cdot)$ by at most an additive constant (invariance theorem), which becomes a multiplicative constant on $m(\cdot)$.
\item Theorems~\ref{thm:coding-uncond}–\ref{thm:coding-cond} are often summarized as
$m(x)\asymp 2^{-K(x)}$ and $m(x\mid y)\asymp 2^{-K(x\mid y)}$, read ``within constant factors''.
\item Immediate corollaries used in the main text include the \emph{posterior under the universal prior}:
for any program $p$ with $U(p)=x$,
\[
\Pr\{p\mid x\}=\frac{2^{-|p|}}{m(x)}\in\Big[\tfrac{1}{c_2},\tfrac{1}{c_1}\Big]\!\cdot 2^{\,K(x)-|p|},
\]
and the \emph{geometric excess‑length tail}:
$\Pr\{|p|\ge K(x)+k\mid x\}\le C\,2^{-k}$ for some constant $C>0$.
\end{itemize}

\paragraph{References.}
Original sources and standard expositions:
\cite{Solomonoff64a,solomonoffFormalTheoryInductive1964a,ZvonkinLevin70,LiVitanyi4th,Vitanyi2013,Hutter2007}.

\subsection{Why many long descriptions imply compressibility, and why long generators are unlikely}
\label{sec:multiplicity-compression}

Fix a universal \emph{prefix} Turing machine $U$. For a finite binary string $x$,
\[
K(x) \;:=\; \min_{p:\,U(p)=x} |p|
\]
is (prefix) Kolmogorov complexity, and the Solomonoff--Levin a priori semimeasure is
\[
m(x) \;=\; \sum_{p:\,U(p)=x} 2^{-|p|}.
\]
The \emph{coding theorem} (a.k.a.\ Levin's theorem) states that there exist machine‑dependent constants $c_1,c_2>0$ such that
\begin{equation}
\label{eq:coding-theorem}
c_1\,2^{-K(x)} \;\le\; m(x) \;\le\; c_2\,2^{-K(x)}.
\end{equation}
(Background: \href{https://raysolomonoff.com/publications/1964pt1.pdf}{Solomonoff, 1964a}, 
\href{https://raysolomonoff.com/publications/1964pt2.pdf}{1964b}; 
\href{https://www.its.caltech.edu/~matilde/ZvonkinLevin.pdf}{Zvonkin--Levin, 1970}; 
pedagogical survey: \href{https://arxiv.org/abs/1206.0983}{Vit\'anyi, 2013}; 
overview: \href{https://www.hutter1.net/ai/algprob.pdf}{Hutter, 2007}.)

\subsubsection*{Multiplicity $\Rightarrow$ compression (indexing among outputs)}
For $L\in\mathbb{N}$ let $N_{\le L}(x)$ be the number of programs of length $\le L$ that output $x$. 

\begin{lemma}[Multiplicity compression]
\label{lem:multiplicity}
If $N_{\le L}(x)\ge 2^r$, then
\[
K(x) \;\le\; L - r \;+\; O(\log L).
\]
\end{lemma}

\begin{proof}[Proof idea (pedagogical)]
Enumerate all programs of length $\le L$ in dovetailing fashion and record each \emph{distinct} output when first seen; this yields a computable list $\mathcal{A}_L=(x_1,x_2,\ldots)$. 
Define the high‑multiplicity set $\mathcal{B}_{L,r}:=\{x\in\mathcal{A}_L: N_{\le L}(x)\ge 2^r\}$. Each $x\in\mathcal{B}_{L,r}$ ``uses'' at least $2^r$ programs, and the total number of prefix programs of length $\le L$ is $<2^{L+1}$ (Kraft inequality). Hence
\[
|\mathcal{B}_{L,r}| \;\le\; \frac{2^{L+1}}{2^r} \;=\; 2^{L-r+1}.
\]
Therefore $x\in\mathcal{B}_{L,r}$ is specified by: (i) a self‑delimiting code for $(L,r)$ costing $O(\log L)$ bits, and (ii) its index in $\mathcal{B}_{L,r}$ costing $\le L-r+1$ bits. A fixed decoder reconstructs $x$ from these data, yielding the stated bound on $K(x)$.
\end{proof}

\paragraph{One‑line ``weight counting'' variant.}
Since every program of length $\le L$ contributes at least $2^{-L}$ to $m(x)$,
\[
m(x)\;\ge\; N_{\le L}(x)\,2^{-L}
\quad\Rightarrow\quad
N_{\le L}(x)\;\le\; m(x)\,2^{L}\;\le\; c_2\,2^{L-K(x)} 
\;\; \text{by \eqref{eq:coding-theorem}.}
\]
Rearranging gives Lemma~\ref{lem:multiplicity} with the $O(1)$ hidden in constants.

\subsubsection*{Consequences for posterior over program lengths}
Let $N_{b}(x)$ be the number of \emph{exactly} $b$‑bit programs with output $x$. Under the universal prior over programs, $\Pr\{p\}=2^{-|p|}$, observing $x$ induces the posterior
\[
\Pr\{|p|=b \mid x\}
\;=\;
\frac{\sum_{p:\,U(p)=x,\,|p|=b}2^{-|p|}}{m(x)}
\;=\;
\frac{N_b(x)\,2^{-b}}{m(x)}.
\]
Bounding $N_b(x)$ via $m(x)\ge N_b(x)\,2^{-b}$ and \eqref{eq:coding-theorem} gives $N_b(x)\le c_2\,2^{b-K(x)}$. Combining with the lower bound $m(x)\ge c_1\,2^{-K(x)}$ yields the \emph{geometric decay with excess length}:

\begin{theorem}[Excess‑length posterior decay]
\label{thm:geometric}
For all $b\ge K(x)$,
\[
\Pr\{|p|=b \mid x\}
\;\le\;
\frac{c_2}{c_1}\,2^{-(\,b-K(x)\,)}.
\]
Equivalently, writing $b=K(x)+k$ with $k\ge 1$,
\[
\Pr\{|p|=K(x)+k \mid x\} \;\le\; C\,2^{-k}
\quad\text{and}\quad
\Pr\{|p|\ge K(x)+k \mid x\} \;\le\; 2C\,2^{-k},
\]
for a machine‑dependent constant $C>0$.
\end{theorem}

\paragraph{Interpretation.}
Every extra bit beyond $K(x)$ halves the posterior mass (up to a constant factor). Thus an observed output $O$ with $K(O)=a$ is \emph{a priori very unlikely} to have been produced by a program $b\gg a$: the posterior probability falls like $2^{-(b-a)}$.

\subsubsection*{Why indexing becomes \emph{shorter} when there are many programs}
The key to Lemma~\ref{lem:multiplicity} is that we index \emph{outputs with many descriptions}, not the descriptions themselves. As the multiplicity $N_{\le L}(x)$ grows by a factor of $2^r$, the set of such outputs shrinks by the same factor, so the index shortens by $r$ bits; this directly yields the $L-r$ bound. (See also exercises and discussion in \href{https://link.springer.com/book/10.1007/978-3-030-11298-1}{Li--Vit\'anyi, 4th ed., Chs.\ 2--3} and an accessible column by \href{https://www.uni-ulm.de/fileadmin/website_uni_ulm/iui.inst.190/Mitarbeiter/toran/beatcs/column94.pdf}{Vereshchagin, 2008}.)

\subsubsection*{Remarks}
(i) Prefix complexity is essential: the domain of $U$ is prefix‑free, giving Kraft's inequality and the well‑defined prior $m(\cdot)$.\; 
(ii) Conditional variants follow verbatim: replace $K(\cdot)$ by $K(\cdot\mid y)$ and $m(\cdot)$ by $m(\cdot\mid y)$ (see \href{https://arxiv.org/abs/1206.0983}{Vit\'anyi, 2013}).\;
(iii) There is no uniform \emph{lower} bound in $k$: for some $x$ there may be no programs of some intermediate lengths due to prefix‑freeness; Theorem~\ref{thm:geometric} gives an essentially tight \emph{upper} bound on the posterior mass at/above length $K(x)+k$.

\medskip
\noindent\textbf{Primary sources with links:}
\quad \href{https://raysolomonoff.com/publications/1964pt1.pdf}{Solomonoff (1964a)}, 
\href{https://raysolomonoff.com/publications/1964pt2.pdf}{Solomonoff (1964b)}, 
\href{https://www.its.caltech.edu/~matilde/ZvonkinLevin.pdf}{Zvonkin \& Levin (1970)}, 
\href{https://arxiv.org/abs/1206.0983}{Vit\'anyi (2013)}, 
\href{https://www.hutter1.net/ai/algprob.pdf}{Hutter (2007)}, 
\href{https://link.springer.com/book/10.1007/978-3-030-11298-1}{Li \& Vit\'anyi (4th ed.)}, 
\href{https://www.uni-ulm.de/fileadmin/website_uni_ulm/iui.inst.190/Mitarbeiter/toran/beatcs/column94.pdf}{Vereshchagin (2008)}.

\subsection{Single-episode compressibility is non-diagnostic} \label{app:non-diagnostic}

 Intuitively, knowing that the regulator-world coupled system produces a low-complexity world output $x$ reduces the set of possible worlds to select from.  In turn, this allows for a shorter description of the world using $R$ and the complexity bound of the output. The program may say: ``To specify $W$, run the dynamics for all possible $W$-$R$ pairs and delete all world model candidates with complex outputs (above the set complexity bound $K(x) < a$). Then use a reduced index to identify $W$". This means that $K(W|R, ``K(x) <a")<K(W)$,  which implies $M(W;R|``K(x) <a") >0$.

\begin{theorem}{(low complexity output $\Rightarrow$ strict but tiny shrinkage)}
Fix a universal prefix machine. Let $m:=|W|$ and $r:=|R|$ denote minimal code lengths, and let $N\ge m$.
For a fixed regulator $R$ and horizon $N$, consider the class $\mathcal{P}_m$ of all \emph{minimal} $m$-bit world programs.
Assume we only know that the closed-loop transcript has \emph{low} complexity,
\[
E_a:\qquad K\!\big(O^{(N)}_{W,R}\mid R,N\big)\ \le\ a,
\]
for some threshold $a<m-c$, where $c=O(1)$ is a machine-dependent constant.
\emph{Claim.} The set of candidates consistent with $E_a$ is a \emph{strict} subset of $\mathcal{P}_m$:
\[
S_{R,N,a}(m)\ :=\ \Big\{\, W\in\mathcal{P}_m\ :\ K\!\big(O^{(N)}_{W,R}\mid R,N\big)\le a \,\Big\}
\ \subsetneq\ \mathcal{P}_m.
\]
Consequently,
\[
K\!\big(W\mid R,E_a\big)\ \le\ \log_2\!\big(|\mathcal{P}_m|-1\big)\ <\ \log_2|\mathcal{P}_m|\ =\ m\ \pm O(1),
\]
i.e., \emph{strictly} $K(W\mid R,E_a)<K(W)$ (by a vanishingly small amount).
\end{theorem} 
\emph{Proof.}
By Kleene’s recursion theorem (quines), there exists a program $W^\star\in\mathcal{P}_m$ that
\emph{prints its own source} as the first $m$ output bits and then halts (or pads). Hence
$K\!\big(O^{(N)}_{W^\star,R}\mid R,N\big)\ge K(W^\star)-O(1)=m-O(1)>a$, so $W^\star\notin S_{R,N,a}(m)$.
Therefore $S_{R,N,a}(m)\subsetneq\mathcal{P}_m$, implying $\log|S_{R,N,a}(m)|<\log|\mathcal{P}_m|=m\pm O(1)$.
\hfill$\square$

To see how small the information gained can be, consider a world program $W$ whose last line is “print $u \times O_R$,” where $u$ is some computed world variable.
If $R$ simply outputs $0$, the world output becomes the all‑zeros string, hence very compressible.
Knowing that $R$ outputs $0$ and that the world output is $0^N$ \emph{does} restrict the structure of the world program (it must include the final multiplication by the regulator output, or something similar on the realized trace), but that restriction can be tiny—the calculation of $u$ may still be arbitrarily complex.

Although we have shown that $R$ and $E_a$ together share information with $W$, it may be very small for any given     case, and, in any case,  this does not imply that $R$ and $W$ share information.   
The chain rule gives
\begin{align*}
M\big(W:(R,E)\big)
&= K(W)+K(R,E)-K(W,R,E)\\
&= K(W)+\big[K(R)+K(E\mid R)\big] - \big[K(R)+K(W,E\mid R)\big] \pm O(\log)\\
&= \underbrace{K(W)+K(R)-K(W,R)}_{M(W:R)}
 \;+\;\underbrace{K(W\mid R)+K(E\mid R)-K(W,E\mid R)}_{M(W:E\mid R)}
 \;\pm O(\log).
\end{align*}
Thus, knowing that the coupled $(W,R)$ system produces a low-complexity readout $x$ in a single run strictly prunes the set of candidate worlds, but in the worst case this shrinkage is only $O(1)$ and---critically---does not by itself imply $M(W{:}R)>0$; it certifies at most $M\!\big(W:(R,E_a)\big)>0$ via the chain rule.

 Does contrast fix the non‑probabilistic identifiability? 
Let $E_{a,b}$ be the (contrastive) event
\[
E_{a,b}:\qquad K\!\big(O^{(N)}_{W,R}\big)\le a\quad\text{and}\quad K\!\big(O^{(N)}_{W,\varnothing}\big)\ge b\ (b>a).
\]
The deterministic shrinkage equals
\[
K(W)-K\big(W\mid R,E_{a,b}\big)\;=\;M\!\big(W:(R,E_{a,b})\big)\ \pm O(\log),
\]
and by the chain rule this splits as
\begin{equation}\label{eq:contrast-chain}
M\!\big(W:(R,E_{a,b})\big)\;=\;M(W{:}R)\;+\;M\!\big(W{:}E_{a,b}\mid R\big)\ \pm O(\log).
\end{equation}
Thus, from single‑episode ON/OFF facts we can certify at most $M\!\big(W:(R,E_{a,b})\big)>0$; in general this does \emph{not} imply $M(W{:}R)>0$, because the conditional term $M(W{:}E_{a,b}\mid R)$ can carry (almost) all the gain or because of \textit{synergy}.

Furthermore, even if the mutual algorithmic information between world and regulator is null, it may be the case that coupling them leads to a reduction of complexity in the world output by chance. 

 These caveats motivate the probabilistic analysis in the paper.

We discuss in more detail the case of synergy, and also show that a decrease of complexity cannot certify mutual information in a particular case.

\subsubsection*{Chain rule and a synergy counterexample.}
By the chain rule for mutual information,
\begin{equation}
  M\!\big(W:(R,E_{a,b})\big)
  \;=\; M(W{:}R) \;+\; M\big(W{:}E_{a,b}\mid R^{\*}\big) \;+\; O(\log n),
  \label{eq:chain}
\end{equation}
where $R^{\*}$ is a shortest description of $R$ (drop $R^{\*}$ and the $O(\log n)$ term in the Shannon case).\footnote{Algorithmic version: \href{https://doi.org/10.1007/978-3-030-11298-1}{Li \& Vitányi, \emph{An Introduction to Kolmogorov Complexity and Its Applications}, 4th ed., Springer, 2019}. Shannon version: \href{https://doi.org/10.1002/047174882X}{Cover \& Thomas, \emph{Elements of Information Theory}, 2nd ed., Wiley, 2006}.}
Thus, observing that $M\!\big(W:(R,E_{a,b})\big)>0$ does not imply $M(W{:}R)>0$, because the conditional term $M\big(W{:}E_{a,b}\mid R\big)$ can carry (almost) all of the gain.

\emph{Example (XOR/synergy).}
Let $R,E_{a,b}\in\{0,1\}^n$ be independent, incompressible strings, and set $W=R\oplus E_{a,b}$ (bitwise XOR). Then:
\begin{align}
  M(W{:}R)
  &\stackrel{+}{=} K(W)+K(R)-K(W,R) \nonumber\\
  &\stackrel{+}{\le} K(W)-K(W\mid R) \nonumber\\
  &\stackrel{+}{=} K(W)-K(E_{a,b}\mid R) \nonumber\\
  &\stackrel{+}{\le} O(\log n),
\end{align}
because $E_{a,b}\mapsto W$ is a bijection given $R$ and independence gives $K(E_{a,b}\mid R)\stackrel{+}{\ge} n$.
In contrast,
\begin{align}
  M\!\big(W:(R,E_{a,b})\big)
  &\stackrel{+}{=} K(W)-K\!\big(W\mid R,E_{a,b}^{\*}\big) \nonumber\\
  &\stackrel{+}{\ge} n - O(\log n),
\end{align}
since $K(W\mid R,E_{a,b})=O(1)$ and $K(W)\stackrel{+}{\ge}K(W\mid R)\stackrel{+}{\ge} n$. Hence the conditional term $M\big(W{:}E_{a,b}\mid R\big)$ carries essentially all the information.
For the Shannon analogue, take $R,E_{a,b}\sim\mathrm{Ber}(\tfrac12)$ i.i.d.; then $I(W;R)=0$, $I(W;E_{a,b}\mid R)=H(W)=n$, so $I\!\big(W;(R,E_{a,b})\big)=n$.\footnote{XOR–synergy as a canonical case in multivariate information: \href{https://arxiv.org/abs/1004.2515}{Williams \& Beer (2010)}. For the identity $M(x{:}y)=K(x)-K(x\mid y^{\*})+O(\log)$ used above, see \href{https://doi.org/10.1109/18.681318}{Bennett, Gács, Li, Vitányi \& Zurek, \emph{IEEE Trans. Inf. Theory}, 1998}.}

\subsection*{Chance simplification with $M(W{:}R)\approx 0$ is possible.}
 
Fix a universal prefix Turing machine $U$ and a finite horizon $N$. All complexities are
implicitly conditioned on $N$ (we write $K(\cdot)$ for $K(\cdot\,|\,N)$). Identify Turing machines
with their shortest prefix codes and write $|W|=K(W)$, $|R|=K(R)$. The coupled world–regulator
system produces a deterministic readout
\[
x \;:=\; O^{(N)}_{W,R}\in\{0,1\}^N.
\]
There is \emph{no} auxiliary map: a fixed, constant‑overhead wrapper decodes $(W,R)$ and
simulates the interaction to print $x$ (decode+simulate). Consequently
\begin{equation}
K(x)\;\le\;K(W,R)+O(1)\;=\;K(W)+K(R)-M(W{:}R)\pm O(\log N),
\label{eq:joint-upper}
\end{equation}
and we use the standard identity $M(X{:}Y)=K(X)-K(X\mid Y^{\*})\pm O(\log)$.
(See eq. (1) and the chain‑rule algebra in §2–3 of the WP.)\footnote{For textbook
background on prefix complexity, chain rules, and $M(x{:}y)=K(x)-K(x\mid y^{\*})\pm O(\log)$,
see \href{https://link.springer.com/book/10.1007/978-3-030-11298-1}{Li \& Vitányi (2019)},
and \href{https://doi.org/10.1109/18.681318}{Bennett et\,al.\ (1998)}.}

For concreteness in the examples below we take $|W|=|R|=n$ and set $N=n$; this is only for
clarity (all statements have the obvious adjustments if $N\neq n$).

\paragraph{Claim (It can happen that $K(x)$ is small while $M(W{:}R)\approx 0$).}
There exist pairs $(W,R)$ with $M(W{:}R)=O(\log n)$ such that the coupled output $x=O^{(N)}_{W,R}$
has very small complexity (e.g. $K(x)=O(\log N)$).

\emph{Construction (existence, uses only the $W{+}R$ coupling).}
Fix a threshold $\Delta\in\{1,\dots,N\}$. Define a world program $W_{\Delta}$ that monitors the first $\Delta$
symbols emitted by the regulator on the interface and then latches:
\[
\text{if }O_R[1{:}\Delta]=0^{\Delta}\text{ then output }x=0^N;\quad\text{else output a fixed incompressible }z\in\{0,1\}^N.
\]
Here $z$ is hard‑coded in $W_{\Delta}$ (so $K(z)\stackrel{+}{=}N$ and $K(W_{\Delta})\stackrel{+}{=}|W|=n$).
Choose any regulator $R^{(\Delta)}$ whose first $\Delta$ interface outputs are $0^\Delta$ and whose remaining
behavior is generated by a shortest program of length $\stackrel{+}{=}n$ independent of $W_{\Delta}$.
Then
\[
M(W_{\Delta}{:}R^{(\Delta)})=O(\log n)\qquad\text{but}\qquad x=0^N\;\Rightarrow\;K(x)=O(\log N).
\]
Thus, even with $M(W{:}R)\approx 0$ (up to the usual $O(\log)$ slack), the \emph{coupled} program
can, on the realized episode, yield a low‑complexity output.

\paragraph{“Rare but possible” bound (balanced couplings).}
Suppose the world implements a balanced dependence on the regulator’s interface in the sense that,
for fixed $W$, the map $u\mapsto x$ is a permutation of $\{0,1\}^N$ when we view $u:=O_R[1{:}N]$ as
the regulator’s output sequence (e.g., the world computes $x=z\oplus u$ with a fixed $z=z(W)$).
If $R$ is sampled independently and its interface sequence $u$ is (close to) uniform on $\{0,1\}^N$
(e.g., drawn from a family with pseudorandom outputs), then by the standard Kolmogorov counting bound
(at most $2^{k+1}$ $N$‑bit strings have $K\le k$),
\[
\Pr\big[K(x)\le k\big]\;\le\; 2^{\,k+1-N}.
\]
Equivalently, the probability of a $\Delta$‑bit drop ($K(x)\le N-\Delta$) is $\le 2^{\,1-\Delta}$.
Thus, a very simple $x$ can occur \emph{by chance}, but only with exponentially small probability in
the amount of simplification.\footnote{Counting bound: at most $2^{k+1}$ strings of length $N$ have
complexity $\le k$; see \href{https://link.springer.com/book/10.1007/978-3-030-11298-1}{Li \& Vitányi (2019)}.}

\paragraph{Ex‑post constraint when $R$ is invertible from $(W,x)$.}
If the coupled architecture allows recovery of $R$ from $(W,x)$ via a computable inverse (i.e.,
there exists a fixed decoder such that $R=G(W,x)$), then
\[
K(x)\;\ge\;K(R\mid W^{\*})-O(1)\;=\;K(R)-M(W{:}R)-O(\log n).
\]
Hence, with $K(R)\stackrel{+}{=}n$ and $M(W{:}R)\approx 0$, a large drop in $K(x)$ cannot occur under such
invertible (in $R$) couplings. When a very small $x$ is observed in this case, it \emph{forces}
$M(W{:}R)$ to be large. (Identity used: $M(X{:}Y)=K(X)-K(X\mid Y^{\*})\pm O(\log)$.)

\end{document}